\documentclass[11pt]{article}

\usepackage{grffile}

\usepackage{amsmath,amsfonts,amscd,a4wide}
\usepackage{amsmath}
\usepackage{amssymb}
\usepackage{graphicx}
\usepackage{amsthm}
\usepackage{color}

\usepackage{transparent}
\usepackage[skip=10pt,font=footnotesize]{caption}
\usepackage{subcaption}
\usepackage{enumitem}

\newtheorem{Lemma}{Lemma}[section]
\newtheorem{Theorem}{Theorem}

\newtheorem{Remark}{Remark}[section]

\usepackage{hyperref}
\hypersetup{
    colorlinks=true,
    linkcolor=blue,}

\newcommand{\R}{\mathbb{R}}
\newcommand{\C}{\mathbb{C}}

        \newcommand{\tl}[1]{\tilde{#1}}

        \newcommand{\beq}{\begin{equation}}
        \newcommand{\eeq}{\end{equation}}
        \newcommand{\ba}{\begin{align}}
        \newcommand{\ea}{\end{align}}

        \newcommand{\ri}{\mathrm{i}}

        \newcommand{\re}{\mathrm{e}}

 {\begin{trivlist} \item[]{\bf Proof. }}%
 {\hspace*{\fill}$\rule{.4\baselineskip}{.4\baselineskip}$\end{trivlist}}

\newenvironment{Acknowledgement}%
 {\begin{trivlist}\item[]\textbf{Acknowledgements.}}{\end{trivlist}}

\makeatletter\@addtoreset{figure}{section}\makeatother

\makeatletter \@addtoreset{equation}{section} \makeatother

\usepackage[margin=0.75in]{geometry}

\title{ Fronts and patterns with a dynamic parameter ramp} 

\author{
Montie Avery\thanks{Department of Mathematics and Statistics, Boston University, 665 Commonwealth Ave., Boston,  MA 02215, USA;}, 
Odalys Garcia-Lopez\thanks{ Tufts University, Department of Mathematics, 177 College Avenue   
Medford, MA 02155 }, 
Ryan Goh\thanks{Department of Mathematics and Statistics, Boston University, 665 Commonwealth Ave., Boston,  MA 02215, USA;\texttt{rgoh@bu.edu} }, 
Benjamin Hosek\footnotemark[1], 
Ethan Shade\thanks{ University of Colorado - Boulder, Department of Applied Mathematics, 1111 Engineering Center, ECOT 225, 
Boulder, CO 80309}
}

\date{\today}

\begin{document}

\maketitle

\begin{abstract}
We examine the effect of a slowly-varying time-dependent parameter on invasion fronts for which an unstable homogeneous equilibrium is invaded by either another homogeneous state or a spatially periodic state. We first explain and motivate our approach by studying asymptotically constant invasion fronts in a scalar FKPP equation with time-dependent parameter which controls the stability of the trivial state.  Following recent works in the area, we use a linearized analysis to derive formal predictions for front position and leading-edge spatial decay. We then use a comparison principle approach to establish a rigorous spreading result in the case of an unbounded temporal parameter. We then consider patterned-invasion in the complex Ginzburg-Landau equation with dynamic bifurcation parameter, a prototype for slow passage through a spatio-temporal Hopf instability. Linearized analysis once again gives front position and decay asymptotics, but also the selected spatial wavenumber at the leading edge. We then use a Burger's modulation analysis to predict the slowly-varying wavenumber in the wake of the front.  Finally, in both equations, we used the recently developed concept of a space-time memory curve to characterize delayed invasion in the case where the parameter is initially stable before a subsequent slow passage through instability and invasion.  We also provide preliminary results studying invasion in other prototypical pattern formation models modified with a dynamic parameter, as well as numerical results for delayed transition between pushed and pulled fronts in Nagumo's equation with dynamic parameter.

\end{abstract}



\begin{Acknowledgement}
 The authors acknowledge partial support from the NSF through grants DMS-2006887, DMS-2307650  (RG and BH) and DMS-2202714 (MA).
\end{Acknowledgement}

\section{Introduction}

The evolution of coherent structures in the presence of a temporally-dynamic parameter has arisen as a topic of interest in a variety of physical settings. For example, such problems arise when considering striped patterns in a growing or evolving medium \cite{krause2023concentration,tsubota2024bifurcation,madzvamuse2010stability,KRECHETNIKOV201716,goh2023growing}, where the self-similar or apical growth of a domain can, via coordinate change, be represented by a spatio-temporal parameter. Such dynamic parameters also arise in ecological settings, where a model parameter, such as average rainfall, varies and induces phase slip dynamics and wavenumber jumps in patterns \cite{siteur2014beyond}; see also \cite{asch2025slow} for similar studies in a complex Ginzburg-Landau equation.   Other examples of defect formation and supression occur in wrinkled elastic membranes \cite{stoop2018defect}, as well as in fluid and optical systems \cite{casado2006testing,casado07}. In all of the above settings, one generally seeks to understand how the range and slope of the parameter ramp impact pattern characteristics, such as stripe wavenumber and defect distribution. 

While much work has been done on purely periodic phases, little work has been done studying how dynamic parameters interact with pattern-forming fronts which invade an unstable homogeneous state. In static parameter problems, this is a commonly proposed mechanism for selecting a defect-free patterned state with one specific wavenumber \cite{van2003front}. Such fronts arise when an unstable equilibrium is perturbed by compactly supported initial data which grows and spreads, leaving a coherent patterned state in the wake.  In this work, we seek to understand pattern-forming invasion fronts in the presence of a (slowly) evolving parameter. In particular, we to seek determine how the evolving parameter selects the spatial wavenumber at the front interface as well as how the bulk wavenumber evolves in its wake.


We remark there has been a variety of works which characterize how spatio-temporal heterogeneities affect invasion fronts which leave behind a non-patterned, spatially homogeneous state; sometimes referred to as an asymptotically constant front. This includes non-rigorous work studying front position and tail asymptotics \cite{mendez03, tsubota2024bifurcation,ambrosio2021generalized,DUCROT2023} in a variety of models, and a rich literature rigorously considering them in scalar reaction diffusion equations of the form
\begin{equation}
u_t = a(x,t)u_{xx} + f(u,x,t).
\end{equation} 
Under a variety of assumptions - say for example $a = 1$ and $f = f(u,t)$ bounded in $t$  for $u$ values in between the two asymptotic states of the front -  these works  define and characterize the concept of a \emph{generalized transition wave},  a front-like solution (with certain spatial asymptotics) defined for all $t\in \R$, with different asymptotic speeds at $t = \pm \infty$. Such results, in general, rigorously establish front existence, selection and convergence, and invasion properties for various types of nonlinearities and heterogeneities \cite{shen2006traveling,berestycki2007generalized,berestycki2008asymptotic,shen11,berestycki2012generalized,nadin12,rossi2014transition,nadin2015critical,shen2017stability},   

\paragraph{Overview of our results}
In this work, we assume a slowly-varying, and \emph{unbounded}, temporal parameter ramp,
\begin{equation}\label{e:mu}
    \mu(t) = \epsilon t + \mu_0,
\end{equation} 
and study free invasion fronts arising from localized or compactly supported initial conditions, first in the Fisher-KPP (FKPP) equation \eqref{e:fkpp0}, and second in the pattern-forming complex Ginzburg-Landau (CGL) equation  \eqref{e:cgl0}, listed below.

We use the former to explain phenomenon and motivate and explain our approach and introduce ideas.  The equation takes the form of a scalar reaction-diffusion equation
\begin{equation}\label{e:fkpp0}
u_t = u_{xx} + \mu(t) f(u),\qquad f(u) = u - u^2.
\end{equation}
 This equation produces accelerating fronts which are asymptotically constant in space and connect $u = 1$ to $u = 0$ with an increasingly steep interface; see Figure \ref{f:front_comp} and \ref{fig:5e-3}. We use a linear analysis to obtain accurate leading predictions for spreading properties, such as front position and spatial asymptotics of the front profile. 

To obtain more precise predictions for the front interface, we employ the recently developed \emph{space-time memory curve} concept used to characterize delayed onset of Hopf instability in  spatially extended systems \cite{kaper2018delayed,goh2022delayed,goh2024}. In our setting, such delayed bifurcation and front invasion is observed for $0<\epsilon\ll1$ and $\mu_0<0$. Here, a compactly supported initial condition initially decays pointwise while diffusively spreading, as long as $\mu(t)<0$. Then, as opposed to the spatially uncoupled system where the onset of the large amplitude state is symmetric and occurs at time $t = -2\mu_0/\epsilon$, diffusive coupling leads to onset values of $\mu$ which are spatially-dependent $\mu = \mu_{mc}(x).$  Our phenomological results on the delayed invasion of the nonlinear front complement the recent work \cite{jelbart2024charac} which develops geometric blow up and self-similar variable techniques to track solutions with spatially-localized initial data in a neighborhood of $(u,\mu) = (0,0)$, showing that they stay near the spatially homogeneous solutions of the associated slowly-varying reaction kinetics ODE $u' = \mu f(u)$.

As our specific setting (i.e. where $\mu$ is unbounded and negative for $t<0$) has not, to our knowledge, been given a rigorous treatment in the aforementioned rigorous results, we also give a rigorous proof in Theorem \ref{t:sp} below which establishes spreading properties from (one-sided) compactly supported initial data in the time-heterogeneous Fisher-KPP equation \eqref{e:fkpp0}. It confirms that the asymptotic front position is given at leading-order by the linear prediction $x = \left(4t\int_0^t\mu(s)ds \right)^{1/2}$ for $t$ sufficiently large. We believe such a result, while expected given the previous literature, is new due to the unboundedness of the parameter $\mu(t).$ We use comparison principle techniques and develop a novel sub-solution which allows us to characterize the accelerating front and confirm the predicted front position from the aforementioned linear analysis.



To study patterned fronts, we consider the complex Ginzburg-Landau equation with super-critical nonlinearity,
\begin{equation}\label{e:cgl0}
    A_t = (1+i\alpha) A_{xx} + \mu(t) A - (1+i\gamma)A|A|^2,
\end{equation}
a prototypical model for oscillatory instability and Hopf bifurcaton, as well as patterns, in spatially extended domains \cite{aranson2002world,mielke2002ginzburg}.
For $\mu>0$ constant in time, this equation supports periodic wavetrains as well as fronts which connect $A = 0$ ahead of the front interface to a non-constant, locally periodic ``plane-wave" state $A = r e^{i (kx - \omega t)}$; see Figure \ref{f:cgl_wave} for a depiction of a pattern-forming front in \eqref{e:cgl0}.

We remark here that we use a parameter heterogeneity which only multiplies the linear term as it simplifies the frozen-coefficient nonlinear dispersion relation which relates the local amplitude $r$, wavenumber $k$, and temporal frequency $\omega$ of the above plane waves. Further, this nonlinearity has direct connections to other works which investigate dynamic slow passage through a Hopf bifurcation; see references on delayed Hopf bifurcations above.  We do note that similar results would hold for a nonlinearity of the form $\mu(t)(A - (1+i\gamma)A|A|^2)$. Conversely, we mention here that $\mu(t)$ was introduced outside of the nonlinearity in FKPP (Equation \eqref{e:fkpp0}) to keep solutions bounded for all time $t>0$ which simplifies our rigorous result in Section \ref{s:rig}.

In both equations, we use a Green's function analysis of the linearization about the trivial base state to describe the front position, local invasion speed, and leading edge profile of the invasion front. In the case of the CGL equation, \eqref{e:cgl0}, we then extend this analysis to predict the selected temporal frequency and hence local spatial wavenumber at the front interface. We then employ this prediction as a time-dynamic inhomogeneous Dirichlet boundary condition for a Burgers-type modulational analysis which predicts the local wavenumber of the pattern left behind in the wake of the front.



This work is organized as follows: Section \ref{s:kpp-ph} describes the observed phenomena and derives our formal predictions for front behavior in the FKPP equation \eqref{e:fkpp0} while Section \ref{s:dhb-k} studies delayed bifurcation and invasion when $\mu_0<0$. For the pattern-forming CGL equation \eqref{e:cgl0}, Section \ref{ss:cgl_le} describes phenomena and gives formal predictions for both front position, spatial decay, and leading-edge wavenumber. Section \ref{ss:cgl_bulk} then derives and compares a prediction for the bulk wavenumber while Section \ref{s:dhb-cgl} characterizes spatially-dependent bifurcation delay. In Section \ref{s:rig}, we then return to the FKPP equation and state and prove our rigorous spreading result, Theorem \ref{t:sp}. Finally, Section \ref{s:disc} briefly gives preliminary results on and discusses the extension of our approach to other pattern forming systems such as the Swift-Hohenberg and Cahn-Hilliard equations as well slow transitions between pushed and pulled invasion in a time-dynamic Nagumo equation. We also note that source codes used to produce the computational results of this work can be found at the GitHub repository \url{https://github.com/ryan-goh/fronts-and-patterns-temporal-ramp}.


\section{Phenomena and formal predictions - asymptotically constant fronts}\label{s:kpp-ph}

In order to introduce our formal approach, we first consider the FKPP equation \eqref{e:fkpp0} above and derive a leading-order prediction for the front interface location as time evolves. We remark that similar results, albeit with slightly different derivations or in different settings, can be found in \cite{tsubota2024bifurcation,mendez03}. To predict the leading order front position, it suffices to consider the linearized equation 
\begin{equation}\label{e:lfkpp}
v_t = v_{xx} + \mu(t) v.
\end{equation}
Introducing an integrating factor $v = w \exp\left[ \int_0^t \mu(s)ds\right]$ gives that $w$ solves a heat equation, and hence $v$ has the solution form
$$
v(x,t) = \frac{\exp\left[ \int_0^t \mu(s)ds\right]}{\sqrt{4\pi t}}\int_\R \re^{-\frac{(x-y)^2}{4t}} v(y,0)dy.
$$
Similar to \cite{tsubota2024bifurcation}, we take a delta function initial condition $v(x,0) = \delta_0(x)$ to find 
$$
v(x,t) = \frac{\exp\left[ \int_0^t \mu(s)ds\right]}{\sqrt{4\pi t}} \re^{-\frac{x^2}{4t}}.
$$
We track the right-ward spread of this initial condition by fixing a threshold value $u_\mathrm{th} \in (0,1/2)$ and defining the front interface location to be $x_\mathrm{fr}(t) = \inf_{x\in \R_+}\{ u(x,t) < u_\mathrm{th}\}$. Setting $v(x,t) = u_\mathrm{th}$, solving for $x$, and keeping only the leading-order term in $t\gg1$, we obtain,
\begin{equation}\label{e:sig}
x_\mathrm{fr}(t) \approx \sigma(t):=\left( 4 t \int^t_0 \mu(s) ds\right)^{1/2}.
\end{equation}
If $\mu_0 = 0$ we have $\sigma(t) = \sqrt{2\epsilon t^3} = \frac{\sqrt{2}}{\epsilon}\mu(t)^{3/2}$ and, by differentiating, the instantaneous front speed is
\begin{equation}\label{e:ct}
 x_\mathrm{fr}'(t) \approx c(t):=\sigma'(t) = \frac{3}{2}\sqrt{2\epsilon t} = \frac{3}{2}\sqrt{2\mu(t)}.
\end{equation}

We remark that a frozen-coefficient analysis, where one freezes $t$ and derives an instantaneous linearly selected \emph{pulled} invasion speed $c_\mathrm{frz}(t):=2\sqrt{\mu(t)}$ and then defines $x_{f,\mathrm{frz}}(t) = \int_0^t c_\mathrm{frz}(s)ds$, does not accurately capture the front position.


Figure \ref{f:front_comp} (left) depicts the measured front location in direct numerical simulation plotted against the above leading-order prediction $\sigma(t)$, as well as the frozen-coefficient prediction $x_{f,\mathrm{frz}}(t)$. We find good agreement only with the former prediction and that the naive frozen-coefficient prediction underestimates the front position.

\begin{figure}
    \centering
    \hspace{-0.2in}
    \includegraphics[width=0.35\linewidth]{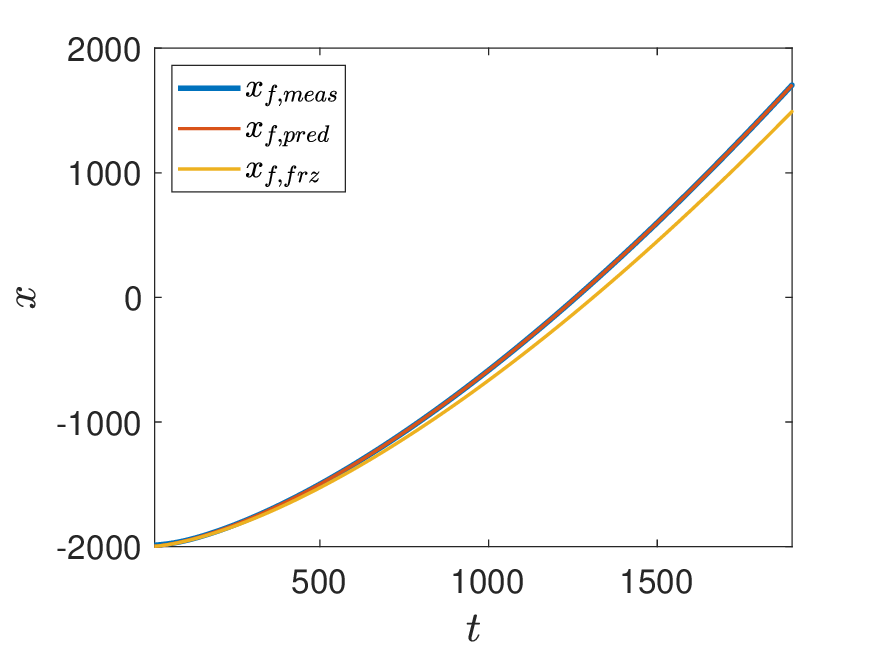} \hspace{-0.1in} 
     \includegraphics[width=0.35\linewidth]{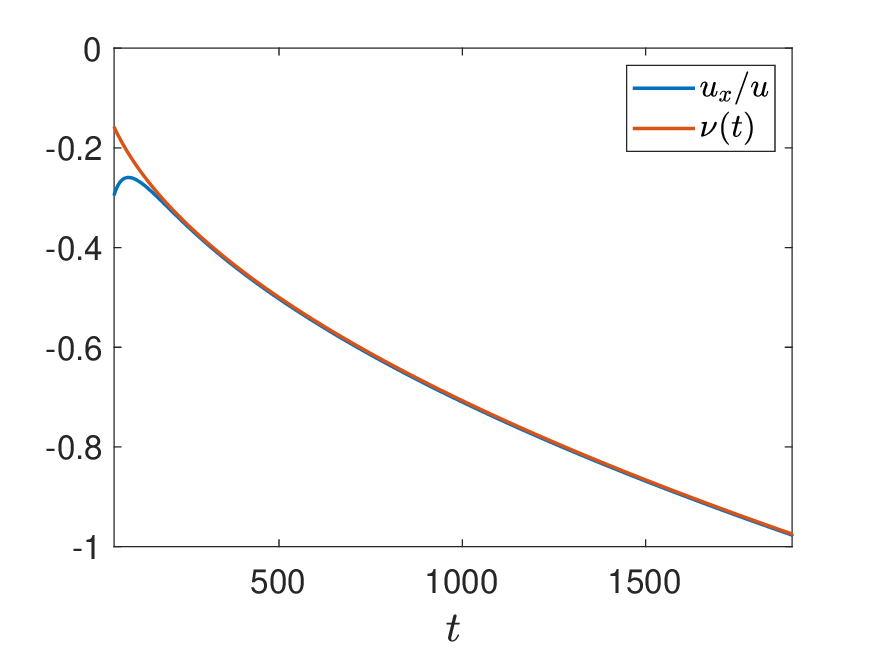}\hspace{-0.2in} 
      \includegraphics[width=0.35\linewidth]{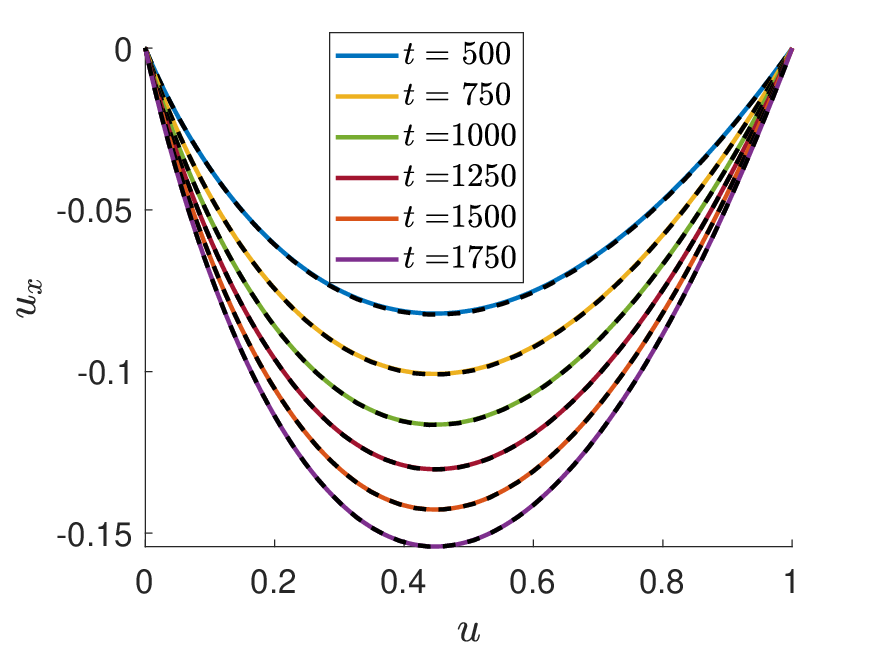}\hspace{-0.2in}
    \caption{Comparison of numerical simulations with predictions. The numerics were performed using 4th-order centered finite-differences in space, and 2nd order balanced Strang splitting \cite{speth13} in time with $dx = 0.005$ and $dt = 0.0025$. Left: A comparison of the numerically measured front position (blue) to both the frozen coefficients prediction (yellow) and the linear prediction (orange). (Center): A comparison of the measured steepness (blue) to the natural asymptotic steepness (orange). (Right): Plot of front profile $u(x,t)$ in $(u,u_x)$ phase plane, for several values of $t$ (given in plot), overlayed with the unique traveling wave trajectory of \eqref{e:tw0} for the same values of $t$ in dashed black.}
    \label{f:front_comp}
\end{figure}

Following \cite{tsubota2024bifurcation}, we can then predict the leading-order spatial asymptotics of the front by transforming into a co-moving frame $z = x - \sigma(t)$, obtaining 
\begin{equation}\label{e:nu}
v(z,t) = \frac{1}{\sqrt{4\pi t}}\exp\left[-\frac{z^2}{4t} - \frac{\sigma(t)}{2t}z\right]
\end{equation}
so that the linear solution, and it turns out the nonlinear solution, has leading order spatial tail $u(z,t)\approx \re^{\nu(t) z}$ with $$\nu(t) = -\sigma(t)/(2t) = -\sqrt{\frac{\epsilon t}{2}} = -\sqrt{\mu(t)/2},$$ indicating that the front steepens as time increases. 

This accurately predicts the steepness of the front found in numerical simulation. Here we measure the time-dependent steepness by computing $u_x(x,t)/u(x,t)$ for $x$ values just ahead of the front interface.
The two are compared in Figure \ref{f:front_comp}(center), where we observe that, after an initial transient as the front establishes itself, the two are in good agreement.
We also remark that the the asymptotic front speed $c(t)$ and front decay rate $\nu(t)$ will be crucial in our rigorous analysis of the nonlinear equation in Section \ref{s:rig}.  We also note that the invasion front exhibits weaker spatial decay than that predicted by the frozen-coefficient analysis since $\nu_\mathrm{frz} = -c_\mathrm{frz}/2 = -\sqrt{\mu} \leq -\sqrt{\mu/2}=\nu <0.$ From point of view of ``envelope velocities" \cite{holzer2014criteria}, this is consistent with the accelerated invasion we observe. When the compact initial data initially spreads, the small $\mu(t)$ prepares a weakly decaying tail, which as $\mu(t)$ increases, causes a instantaneous speed $c(t)$ which is faster than the frozen-coefficient speed  $c_\mathrm{frz}.$

We observe that the nonlinear front profile is well-approximated by the frozen-coefficient traveling wave obtained for each $t>0$ by fixing $\mu = \mu(t), c = \frac{3}{2}\sqrt{2\mu}$ and solving the asymptotic boundary value problem
\begin{equation}\label{e:tw0}
0=u''(z) + cu'(z) + \mu f(u(z)),\qquad \lim_{z\rightarrow-\infty} u(z) = 1,\qquad \lim_{z\rightarrow+\infty} u(z) = 0.
\end{equation}
A comparison of this frozen coefficient traveling wave with numerical simulation is given in Figure \ref{f:front_comp}(right).
This solution is unique up to translations in $z$ and, since the selected speed $c$ is greater than the critical speed $c_{\mathrm{frz}} = 2\sqrt{\mu}$, the spatial profile $u(z)$ has strong exponential decay $u(z)\sim e^{\nu z},\,\, z\gg1$ with $\nu$ as defined above. 

We also remark that this comparison of the front profile in the $u,u_x$ plane indicates that the recently developed ``shape defect function" analysis \cite{an2023quantitative,an2024front} might be of use in characterizing the front; we do not pursue this avenue in this work. We note that \cite{tsubota2024bifurcation} performs a similar analysis albeit for a constant coefficient Ginzburg-Landau equation $A_t = A_{xx} + A - A|A|^2,\,\, A\in \C$ posed in a growing domain $x\in [0, L(t)]$ which, after a time-dependent spatial scaling, yields a equation with time dependent diffusion coefficient.

\section{Higher-order position corrections, space-time memory curves, and delayed invasion}\label{s:dhb-k}

We also explore spatio-temporal delay of invasion by setting $\mu_0<0$ and $0<\epsilon \ll1$ so that the trivial state is initially stable and localized perturbations first decay before beginning to grow after $\mu(t)$ passes through 0 at $t_0 = -\mu_0/\epsilon$.  Without spatial coupling, one would expect a symmetric bifurcation delay, where the solution becomes full amplitude only when $\mu$ reaches roughly $-\mu_0$ (at $t_0 = -2\mu_0/\epsilon$).  In the PDE, the diffusion term induces a spatially dependent delay of bifurcation to the large amplitude state, and hence a delay of spatial invasion in the case of a strongly localized initial condition.  This spatio-temporal delay of spreading can be measured using the linearized analysis above to obtain a curve in the $x,t$ plane which demarcates the pointwise transition to the large amplitude state. In the context of delayed bifurcation, this curve is known as the \emph{space-time memory curve}. It was established in the works \cite{kaper2018delayed,goh2022delayed} and  coined in the recent work \cite{goh2024}. While these works considered spatio-temporal delayed Hopf bifurcation in a CGL equation with an additional symmetry breaking forcing term, such analysis also applies to the FKPP equation \eqref{e:fkpp0} considered above, where the corresponding ODE bifurcation would be a delayed transcritical bifurcation; see also \cite{jelbart2024charac}.  From another point of view, the space-time memory curve gives the precise front position of the linear dynamics, and hence includes the higher-order corrections to the leading-order prediction in \eqref{e:sig}. 

Following \cite{goh2022delayed}, we compute the spacetime memory curves for the FKPP equation with Gaussian initial data $u(x,0)=\exp(-x^2)$, an example initial condition which allows explicit computation. For general initial conditions the curve can be computed by solving an implicit equation or via numerical approximation. Inserting this specific initial condition into the solution formula for the linearized equation \eqref{e:lfkpp} we obtain,
$$
v(x,t) = \frac{1}{\sqrt{4t +1}} \exp\left[ -\frac{x^2}{4t+1} +  \frac{\epsilon t^2}{2} + \mu_0 t \right].
$$

Once again setting $v(x,t) = v_\mathrm{th}$ and solving $v(x,t) = v_\mathrm{th}$, we obtain
\begin{equation}\label{e:memc}
x_\mathrm{mc}(t) = \pm \sqrt{ (4t + 1)\left(  \frac{\epsilon t^2}{2} + \mu_0 t - \log(v_\mathrm{th}) - \frac{1}{2}\log(4t +1)  \right)}.
\end{equation}
Figures \ref{fig:5e-3} left, \ref{fig:e-3} left, and \ref{fig:5e-4} left depict a spacetime diagram of the solution $u$ with $x_\mathrm{mc}$ overlaid in red for $\epsilon = 5\times 10^{-3},1\times 10^{-3},$ and $ 5\times10^{-4}$.  We find good agreement between the linearly predicted space-time memory curve and the nonlinear solution.  We again reiterate that the leading-order front position \eqref{e:sig} can be obtained from \eqref{e:memc} by extracting the leading-order term for $t\gg1.$
\begin{figure}[h!]
    \centering
    \hspace{-0.1in}\includegraphics[width=0.33\textwidth]{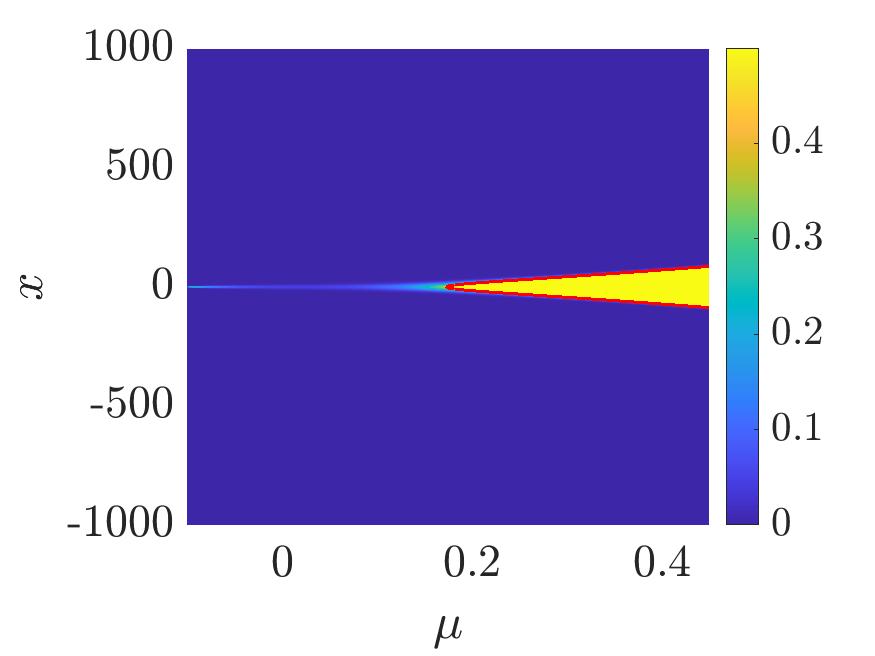}\hspace{-0.05in}	\includegraphics[width=0.33\textwidth]{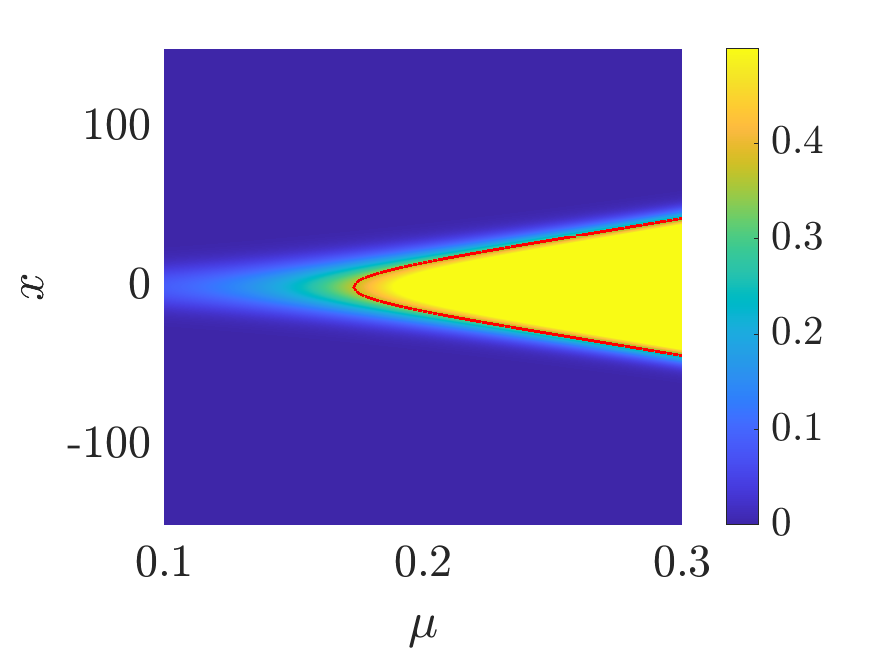}\hspace{-0.05in}
    \includegraphics[width=0.33\textwidth]{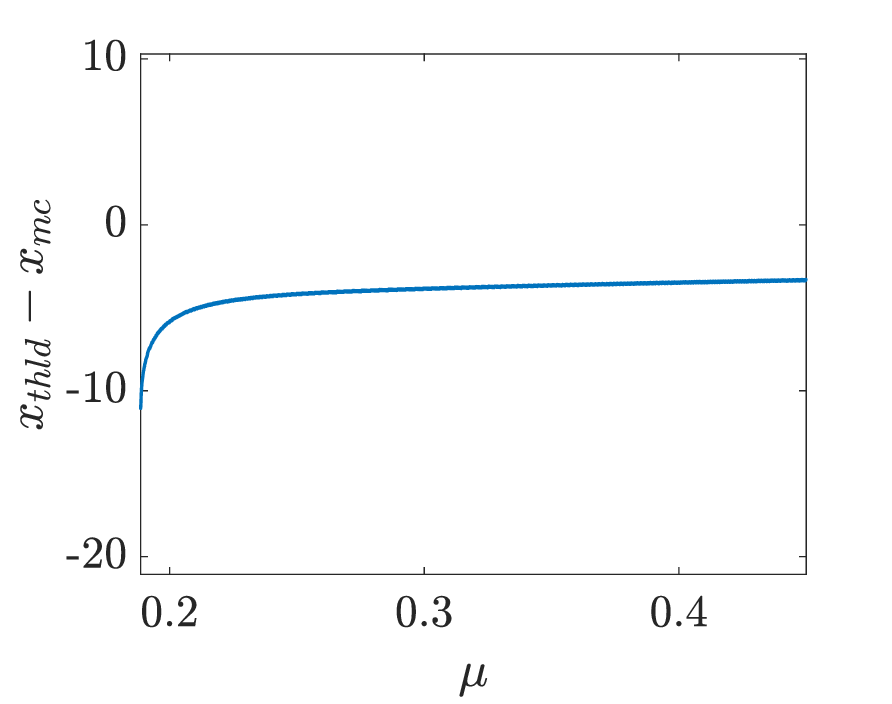}\hspace{-0.1in}
    \caption{$\epsilon =5\times 10^{-3}, u_{th}=0.5,\mu_0=-0.1$.  Left: A numerical simulation of the front position, overlaid by the memory curve.  Center: A closer look around the time at which the front starts spreading.  Right: The error between the numerical measurement and the prediction given by the memory curve.}
    \label{fig:5e-3}
\end{figure}
\begin{figure}[h!]
    \centering
    \hspace{-0.1in}\includegraphics[width=0.33\textwidth]{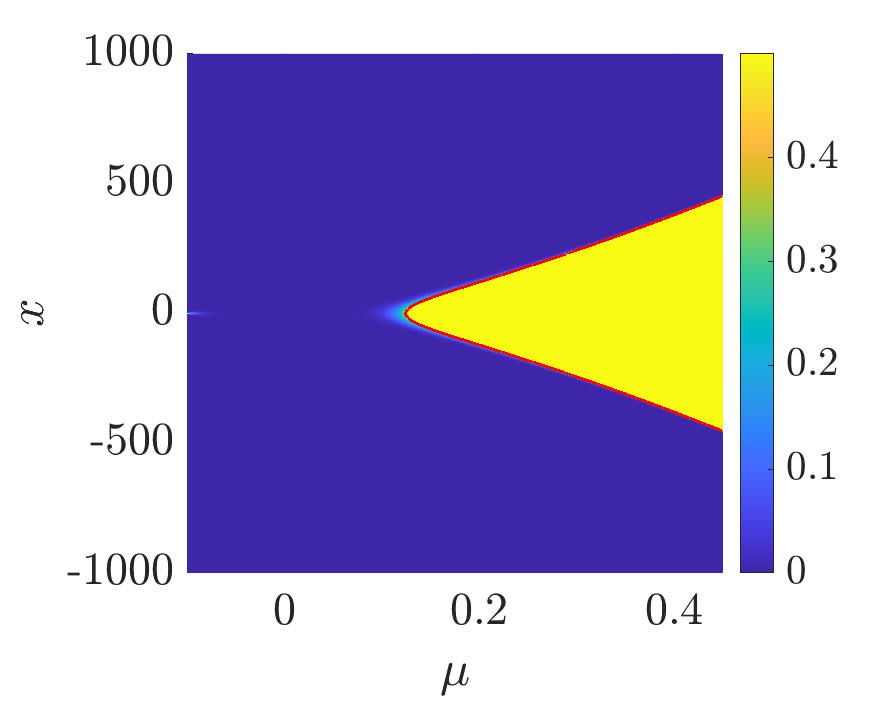}\hspace{-0.05in}\includegraphics[width=0.33\textwidth]{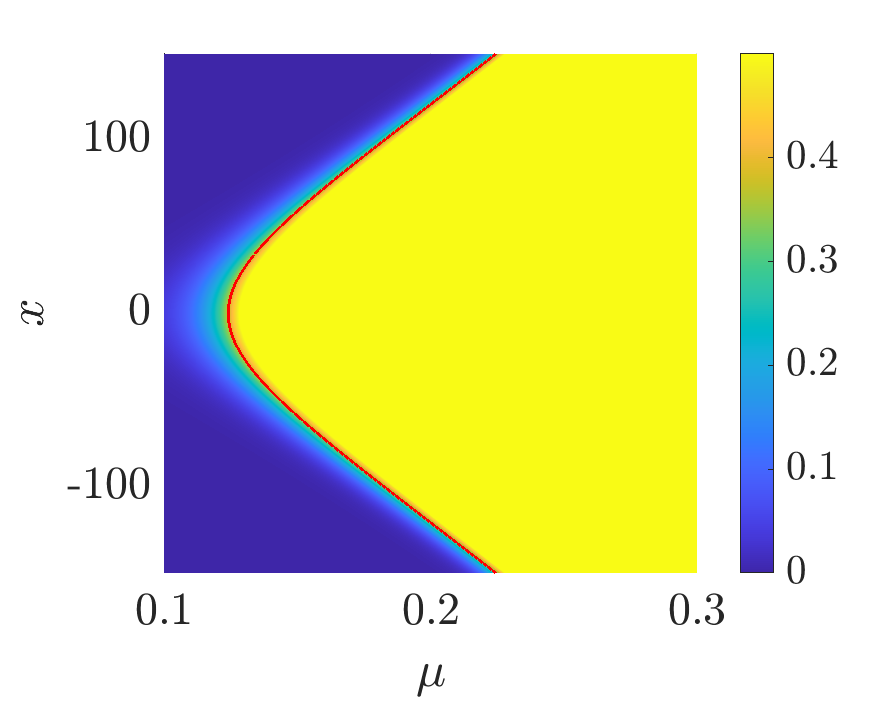}\hspace{-0.05in}
    \includegraphics[width=0.33\textwidth]{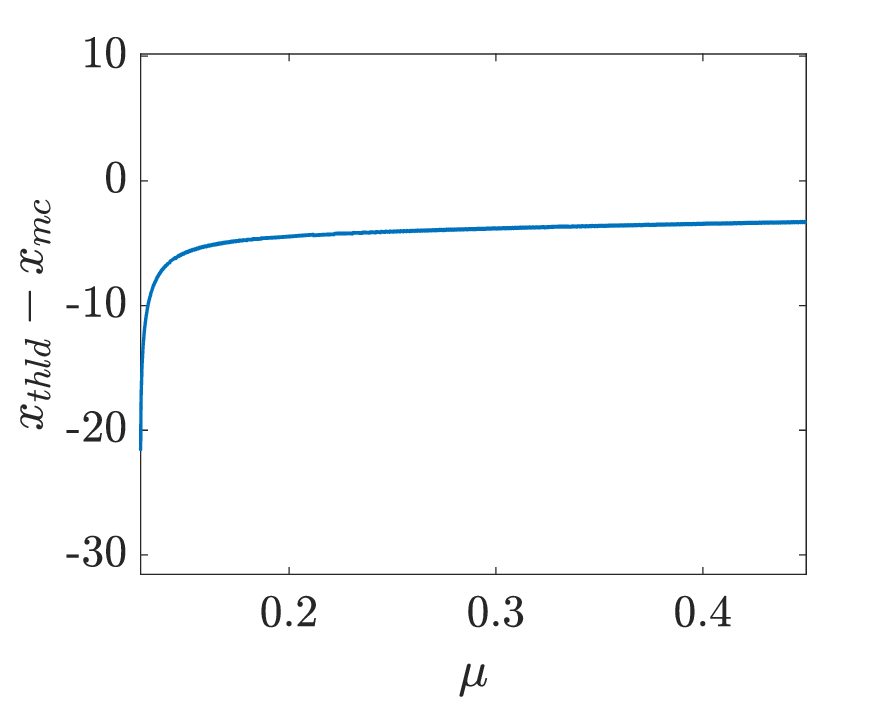}\hspace{-0.1in}
    \caption{$\epsilon =1\times10^{-3}, u_{th}=0.5,\mu_0=-0.1$.  Same plots as in Fig. \ref{fig:5e-3}.}
    \label{fig:e-3}
\end{figure}
\begin{figure}[h!]
    \centering
    \hspace{-0.1in}\includegraphics[width=0.33\textwidth]{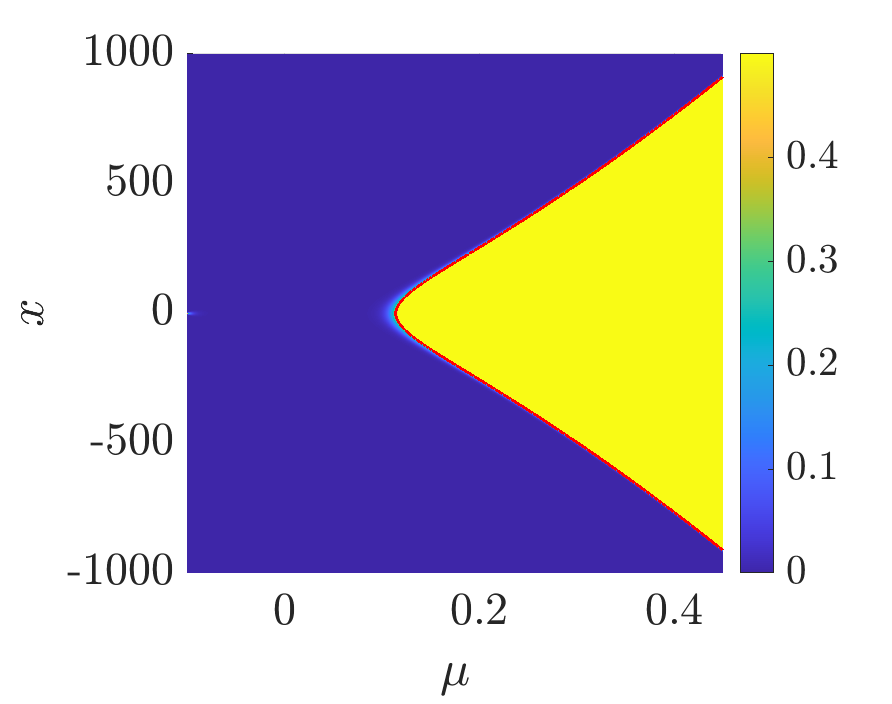}\hspace{-0.05in}\includegraphics[width=0.33\textwidth]{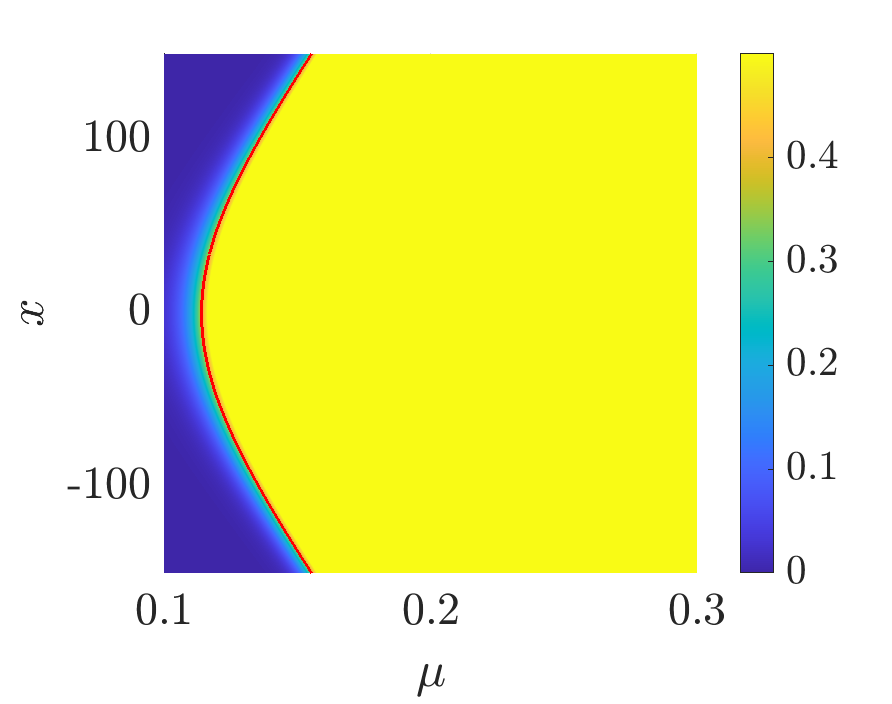}\hspace{-0.05in}
    \includegraphics[width=0.33\textwidth]{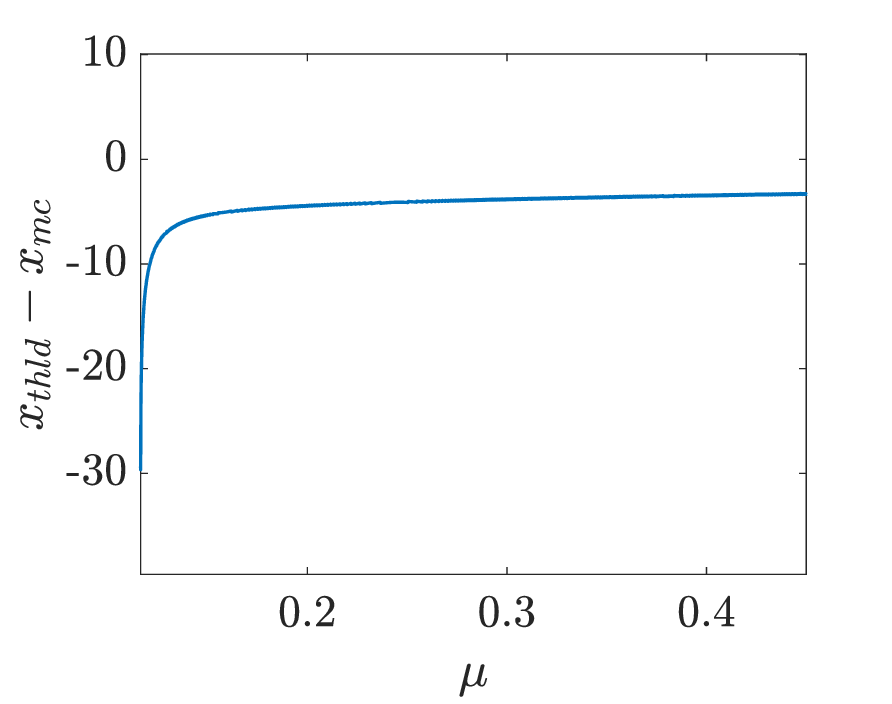}\hspace{-0.1in}
    \caption{$\epsilon =5\times 10^{-4}, u_{th}=0.5,\mu_0=-0.1$. Same plots as in Fig. \ref{fig:5e-3}} 
    \label{fig:5e-4}
\end{figure}
In the right plots of Figures \ref{fig:5e-3}, \ref{fig:e-3}, and \ref{fig:5e-4} we show the error between the numerically measured front position $x_f$ and the prediction given by $x_{\mathrm{mc}}$.  We note that in all cases, this difference is negative, implying that, for a fixed $x$ location, the nonlinear solution becomes large amplitude after the linear solution, indicating a further (likely $\epsilon$-dependent) temporal delay. 
Also in all cases, the discrepancy appears to be approaching a constant shift, uniform in $\epsilon$, despite smaller values of $\epsilon$ causing more significant initial differences.

Focusing on the solution at $x = 0$, we note that spatial diffusive coupling induces a delay in large-amplitude growth compared with the homogeneous symmetric exit time computed from the homogeneous linear ODE.  The center plots of Figures \ref{fig:5e-3}, \ref{fig:e-3}, and \ref{fig:5e-4} all have their range in $t$ starting from the symmetric exit time, i.e., the left boundary of the zoomed-in image is $t_0 = -2\mu_0/\epsilon$.   It is clear from these that the symmetric exit time, is not representative of either the spacetime memory curve or the numerically measured solution.  


\section{Phenomena and formal predictions - pattern-forming fronts}\label{s:cgl}

We now consider invasion in the CGL equation \eqref{e:cgl0} with time-dependent linear parameter $\mu$. Evolving once again from a localized initial data, we observe the same accelerating front interface as $\mu$ increases, but now the local phase at the leading-edge oscillates with increasing frequency as time moves forward.  This oscillatory tail, via the nonlinear dispersion relation for periodic plane waves, establishes a local spatial wavenumber just behind the interface. This oscillatory state then becomes large-amplitude and then mixes with the bulk. In sum, the accelerating front leaves behind a large-amplitude, locally periodic state with slowly-varying amplitude and wavenumber; see Figure \ref{f:cgl_wave}.


\begin{figure}[h!]
    \centering
    \includegraphics[width=0.6\linewidth,trim={0 0 1.5cm 0.0in},clip]{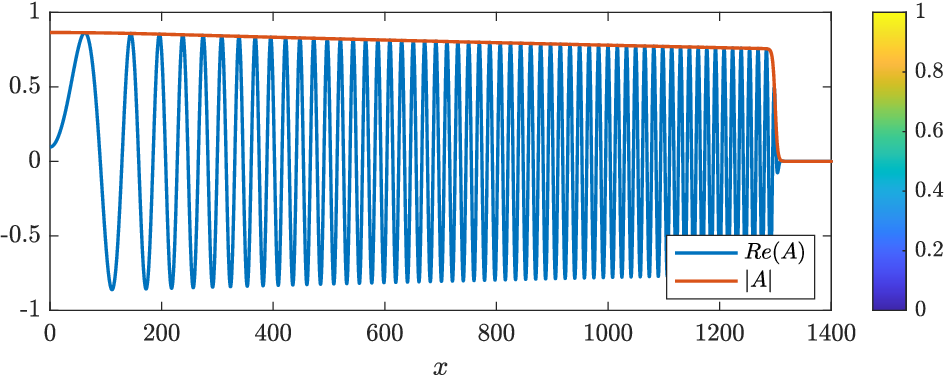}\\    
   \hspace{-0.2in} \includegraphics[width=.525\textwidth]{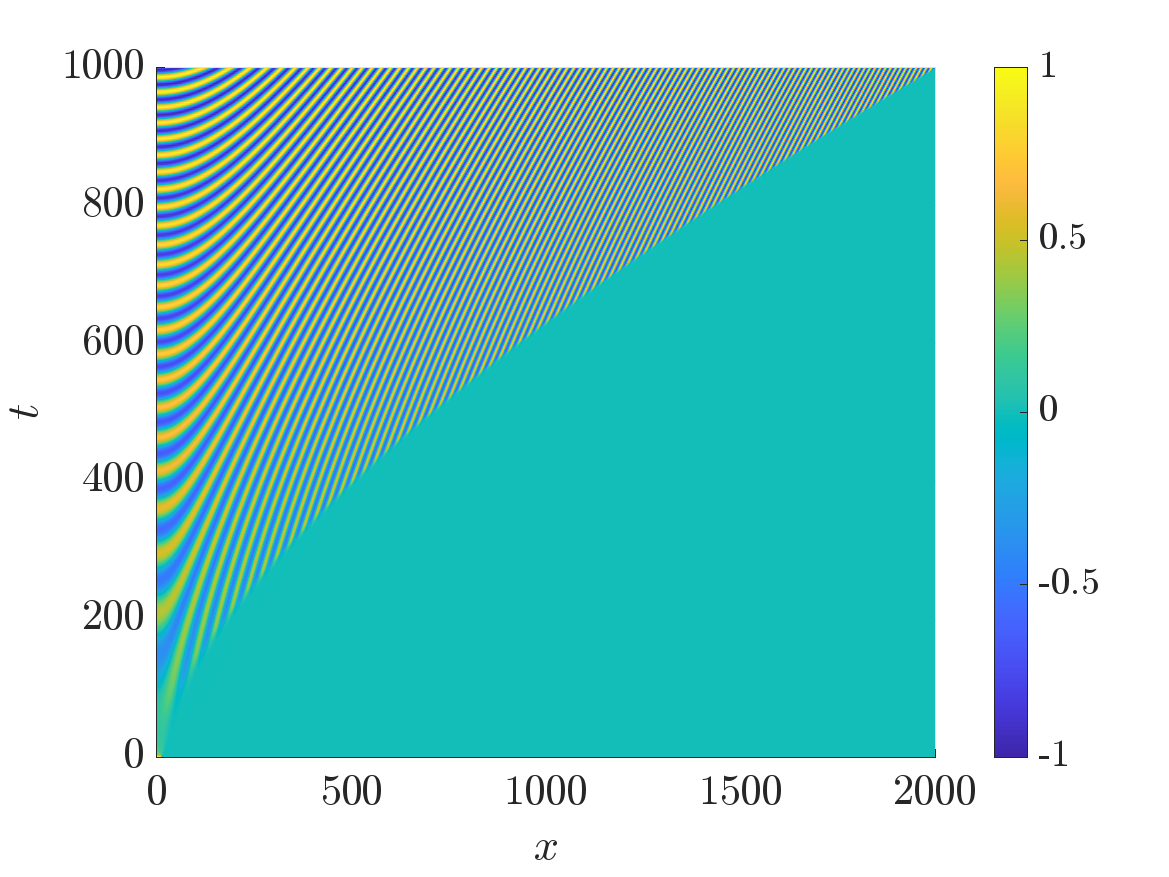}\hspace{-0.25in}
    \includegraphics[width=.525\textwidth]{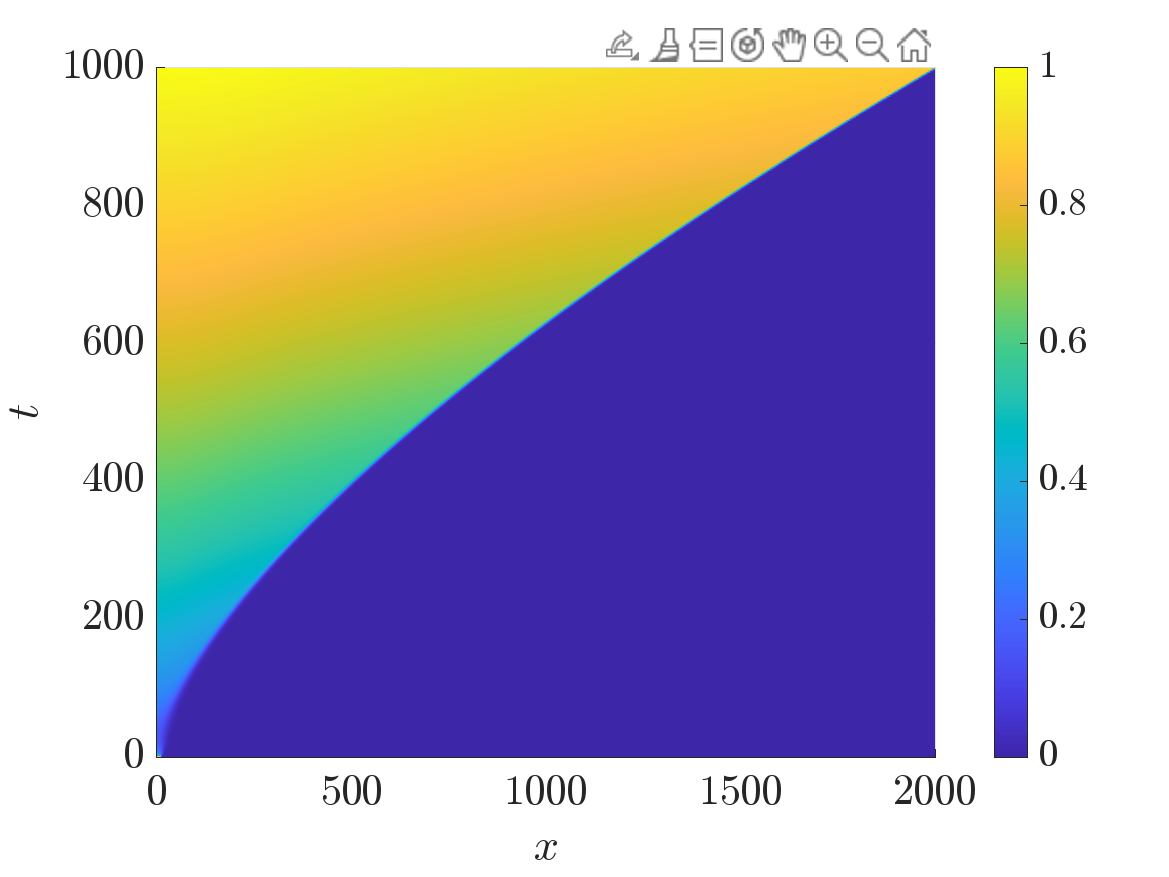}\hspace{-0.2in}
    \caption{Top: A pattern-forming front produced by the complex Ginzburg-Landau equation with an increasing, slowly-varying parameter.  Here $t = 1000$, and $\alpha=1,\gamma=0.3,$ and $\epsilon=0.001$, which are the parameters used throughout this section. We used spatial and temporal discretizations of $dx=0.1, dt=0.0025$ and the same numerical algorithm as used for FKPP. Bottom left: Spacetime-diagram of $\mathrm{R}\, A$; Bottom Right: Space-time diagram of $|A|$. }
    \label{f:cgl_wave}
\end{figure}

\subsection{Leading-edge front location and wavenumber prediction}\label{ss:cgl_le}
We once again consider the linearized $B_t = (1+\ri\alpha) B_{xx} + \mu(t) B$, evolving from delta function initial condition. We obtain a leading order prediction for the front location by solving the threshold equation for the amplitude $|B|$. In more detail, the linear solution for delta function initial condition is given by
$$
B(x,t) = \frac{1}{\sqrt{4\pi(1+\ri\alpha)t}}\exp\left[ \int_0^t\mu(s) ds - \frac{x^2}{4(1+\ri\alpha )t} \right].
$$
Once again assuming $\mu_0 = 0$ and solving  $|B(x,t)| = u_\mathrm{th}$ for $x$ in terms of $t$ at leading order gives the interface prediction
\begin{align}
\sigma(t) = \sqrt{2\epsilon(1+\alpha^2) t^3}.
\end{align}
We then determine the local temporal oscillation frequency at the interface by evaluating the imaginary part of the solution at $x = \sigma(t)$ and computing the local phase 
\begin{align}
\phi_\epsilon(t)&:= \mathrm{Im} \log B(\sigma,t) = \frac{\alpha\sigma(t)^2}{4(1+\alpha^2)t}\notag\\
&= \frac{\alpha \epsilon t^2}{2}. 
\end{align}
The local frequency at the front interface is then
\begin{align}
 \omega_\epsilon(t):= \phi_\epsilon'(t) = \alpha \epsilon t.
\end{align}

The local spatial wavenumber can then be obtained using the frozen coefficient nonlinear dispersion relation posed in a co-moving frame $z = x - ct$, with frozen $c=c(t):=\sigma'(t)$. That is, fixing $\mu,c>0$ in the full nonlinear equation \eqref{e:cgl0}, and inserting the wave-train ansatz $A = r \re^{\ri(k(x - ct) + \omega t)}$, one obtains
\begin{align}\label{e:om0}
\omega = (\gamma - \alpha) k^2 + c k - \gamma \mu,
\end{align}
which has the solutions
$$
k =\left[ -c \pm \sqrt{c^2 + 4(\gamma - \alpha)(\omega + \gamma\mu)} \right]/(2(\gamma - \alpha)),\quad \alpha\neq\gamma.
$$
Then, for each $t>0$, we set $\omega = \omega_\epsilon(t),\, c = c(t)=\sigma'(t),\,$ and $\mu = \mu(t)$, and assume $\gamma - \alpha\neq0$ in this formula to obtain the local wavenumber prediction
\begin{align}
k_\epsilon(t)&= \frac{\sqrt{\mu(t)}}{2\sqrt{2}(\gamma - \alpha)} \left( \sqrt{9+\alpha^2 + 8\gamma^2}-3\sqrt{1+\alpha^2} \right).\label{eq:k_eps}
\end{align}
where we choose signs of the square roots such that $\mu - k^2 >0.$

\begin{figure}[h]
    \centering
   \hspace{-0.25in} \includegraphics[width=0.33\linewidth,trim={0 0 1.75cm 0},clip]{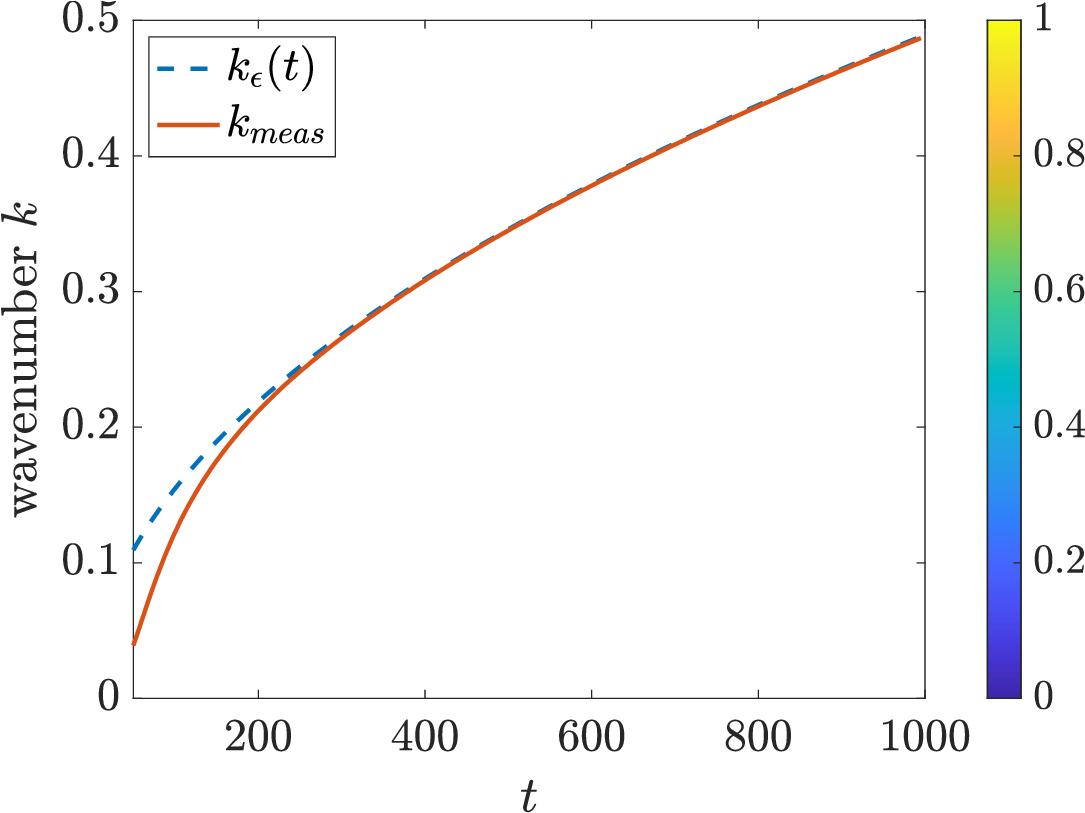}
        \includegraphics[width=0.33\linewidth,trim={0 0 1.7cm 0},clip]{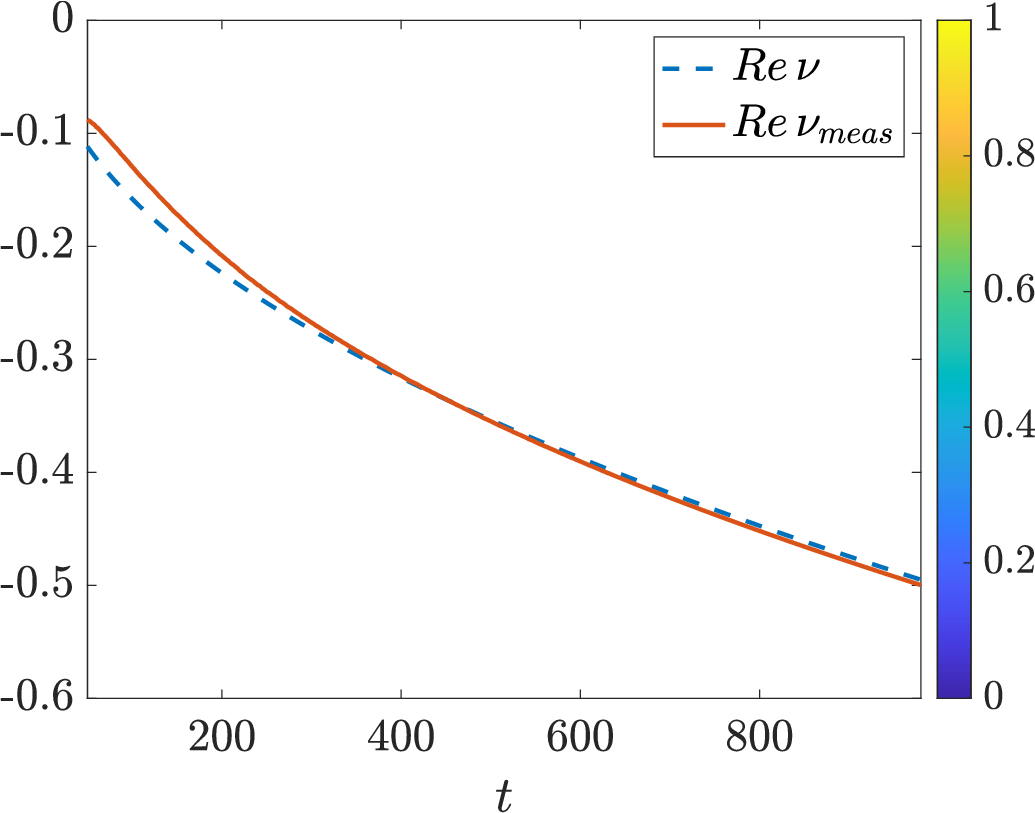}
            \includegraphics[width=0.35\linewidth,trim={0 0 1.7cm 0},clip]{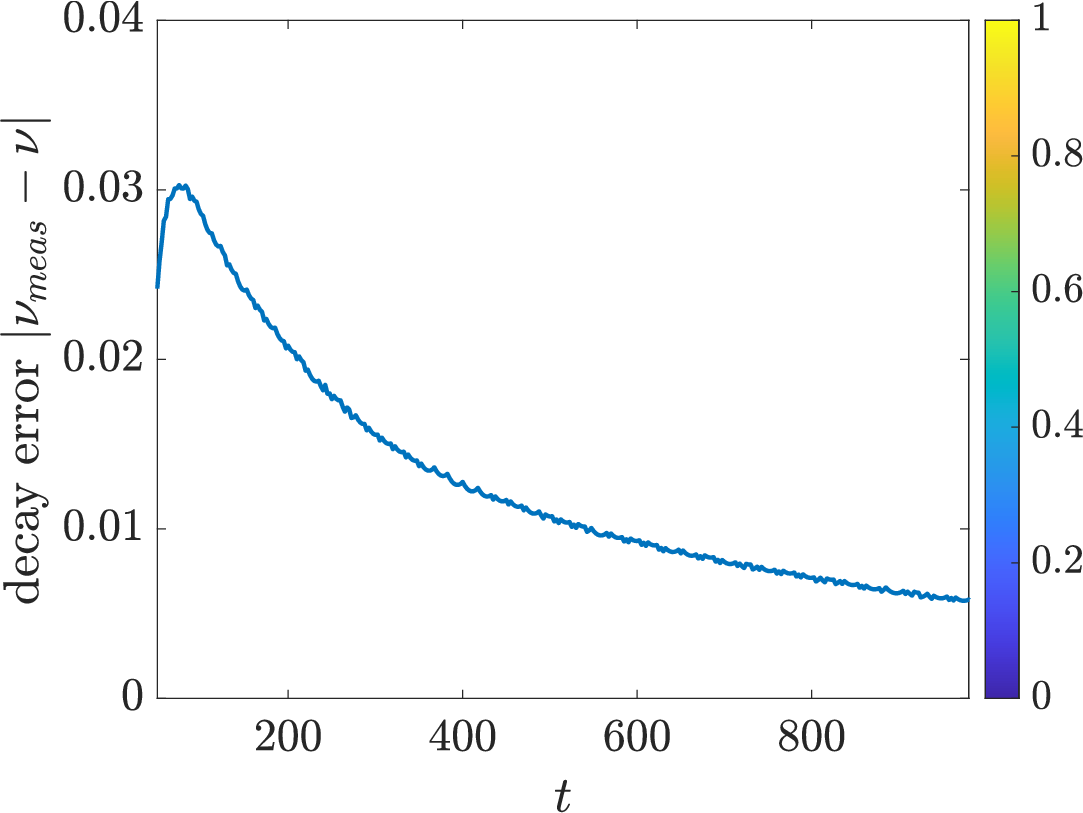} \hspace{-0.25in}
    \caption{Left: Plot of wavenumber measurement $k_\mathrm{meas} = \mathrm{Im}\, A_x/A$ (solid orange) measured just behind the front interface $x = \sigma(t)$ compared with prediction $k_\epsilon(t)$.  Center: Comparison of the decay rate measurement $\mathrm{Re}\, \nu_{meas} = \mathrm{Re}\, A_x/A$ (solid orange), measured just ahead of the front interface, against the prediction  $\nu$ given in \eqref{e:nu_cgl}; Right: Plot of the absolute error $|\nu_{meas} - \nu|$ between measurement and prediction.  }
    \label{f:cgl_lewv}
\end{figure}
Figure \ref{f:cgl_lewv} depicts a comparison of the leading-order prediction, $k_\epsilon$, with the numerical wavenumber measured just behind the front interface. Here, the wavenumber is computed using $k_\mathrm{meas} = \mathrm{Im}\, A_x/A$ for an $x$ value just behind the front interface.  Note that it takes some time for the front and patterned state to fully establish itself, and we find the measured wavenumber converges close to the prediction by time $t\approx 200$. 

The spatial decay profile of the front can also be obtained from the linear solution $B(x,t)$ as done for FKPP above. Namely, one transforms into a co-moving frame $z = x - \sigma(t)$, and expands the argument of the exponential in $B(z+ \sigma(t))$. The factor of the argument of the exponential which is linear in $z$ gives the exponential decay profile $e^{\nu(t) z}$ with 
\begin{equation}\label{e:nu_cgl}
    \nu(t) = -\frac{\sqrt{2\epsilon(1+\alpha^2) t}}{2(1+i\alpha)}.
\end{equation}
Figure \ref{f:cgl_lewv} (center and right) shows that, after an initial transient, the spatial decay envelope is accurately predicted by $\nu(t).$ The decay rate is numerically measured by evaluating $A_x/A$ just ahead of the front interface. The center plot depicts the real part of the measured and predicted decay rates while the right plot depicts the absolute difference between them.

\subsection{Bulk wavenumber modulation equation}\label{ss:cgl_bulk}

To derive a prediction for wavenumber behaviors in the bulk, we use the modulational approach of \cite{howard77} which derives an inviscid Burger's equation to predict the leading-order approximate behavior of the slowly-varying amplitude and phase modulations of a plane wave solution. This approach uses a uniform scaling in space and time $(X,T) = (\delta x, \delta t)$ with $0<\delta \ll1$ and gives validity on the time scale $t\sim 1/\delta.$  We remark that a parabolic scaling $(X,T) = (\delta x, \delta^2 t)$ would give a viscous Burger's modulation equation and larger interval of validity in time, $t\sim 1/\delta^2$. Since the inviscid modulational analysis yields accurate predictions, we do not pursue the latter here; see \cite{doelman2009modulation} for more discussion on both. 


We once again consider frozen coefficient plane wave solutions  $A(z,t)=r\exp(i(kz+\omega t))$ in the co-moving frame $z = x - ct$ where $c$ and $\mu$ are frozen so that $\omega$ and $k$ satisfy \eqref{e:om0} and $r = \sqrt{\mu - k^2}$. Taking into account that $\mu$ and $c$ are time dependent, we then define
\begin{equation}\label{e:omt}
\omega(k,t) = (\gamma - \alpha) k^2 + c(t) k - \gamma \mu(t),\qquad r(k,t) = \sqrt{\mu(t)-k^2}.
\end{equation}
Since solutions initially begin with zero wavenumber, we modulate the $k=0$ mode with the ansatz
 \begin{equation}
    A(z,t)=(r(0,t)+\tilde{r}(z,t))\exp(i(\omega(0,t)+\phi(z,t))).
\end{equation}

Inserting this into the complex Ginzburg-Landau equation gives a system of two coupled PDEs for $\tilde{r}$ and $\phi$.  We define the local wavenumber as $\psi=\phi_z,$ and define the slow variables $(Z,T)=(\delta z,\delta t)$, and $(W,q)(Z,T)=(\tilde{r},\psi)(z,t)$, for $\delta>0$ small enough.  Removing common factors of $\delta$ and then looking at only leading-order terms in $\delta$ gives an algebraic equation which can be solved for $W$.  Substituting this into the remaining equation gives the inviscid Burger's equation for $q$. We represent the front phenomenon in question by imposing a time-dependent Dirichlet boundary condition fixing the wavenumber to be the leading-edge prediction $k_\epsilon$ given in \eqref{eq:k_eps} above.  In sum, we obtain the following modulation equation for the wavenumber $q(Z,T)$ (see Chapter 6 of \cite{doelman2009modulation}), posed on the left half-plane with time-varying boundary condition on the right boundary, and zero initial data $q(Z,0) = 0$:
\begin{align}\label{e:qt}
    q_T-\partial_Z\omega(q,T/\delta)&=0, \qquad\qquad\qquad  Z<0 \notag\\
    q &= k_\epsilon(T/\delta),\qquad Z =0
\end{align}
Note that $\omega$, defined in \eqref{e:omt}, is a function of $q$ and $t$ and $\partial_Z \omega(q,T/\delta) = \omega_q(q,T/\delta)q_Z$. As \eqref{e:qt} is a first-order nonlinear equation, it can be solved explicitly using the method of characteristics. Non-zero characteristics emanate from the right boundary at $Z = 0$. This causes characteristic lines of varying slope and hence a slowly-mixing wavenumber profile in the bulk. 

In more detail, \eqref{e:qt} gives the following system of characteristic equations 
\begin{align}
    \frac{dq}{ds}&=0\\
    \frac{dZ}{ds}&=-\omega_q(q,T/\delta)\\
    \frac{dT}{ds}&=1,
\end{align}
which can be explicitly solved
\begin{align}
    q(s)&=q_0=k_\epsilon(T_0/\delta),\label{e:qs}\\
    T(s)&=s+T_0,\\
    Z(T)&=-2(\gamma-\alpha)C_{\alpha,\gamma,\epsilon}\sqrt{T_0/\delta}(T-T_0)-\sqrt{2\epsilon(1+\alpha^2)T^3/\delta}+\sqrt{2\epsilon(1+\alpha^2)T_0^3/\delta},\label{e:ZT}
\end{align}
where\begin{align}
    C_{\alpha,\gamma,\epsilon}&=\sqrt{\epsilon}\left(\sqrt{9+\alpha^2+8\gamma^2}-3\sqrt{1+\alpha^2}\right)/\left(2\sqrt{2}(\gamma-\alpha)\right).
\end{align}
From this one can solve for $T_0 = T(0)$ in terms of $(Z,T)$, and insert into \eqref{e:qs} to obtain an explicit formula for the wavenumber in terms of the slow variables, after which one can translate back to the original variables $(x,t)$ for comparison with the measured wavenumber. 

First, we note that since $z=x-\sigma(t)$ with $\sigma(t)=\sqrt{2\epsilon(1+\alpha^2)t^3}$, we can write $Z=\delta z=\delta x-\delta\sigma(t)=X-\delta\sigma(T/\delta)$ in \eqref{e:ZT}
to, after some algebraic manipulation, obtain the equation
\begin{equation}
  0 =   h(T_0;T,X):=\sqrt{T_0}^3-\frac{2(\gamma-\alpha)C_{\alpha,\gamma,\epsilon}T}{2(\gamma-\alpha)C_{\alpha,\gamma,\epsilon}+\sqrt{2(1+\alpha^2)\epsilon}}\sqrt{T_0}-\frac{\sqrt{\delta}X}{2(\gamma-\alpha)C_{\alpha,\gamma,\epsilon}+\sqrt{2(1+\alpha^2)\epsilon}}.
\end{equation}
Defining coefficients $b_0(X,T)$ and $b_1(X,T)$ to be such that $h(T_0;T,X) = T_0^{3/2}+b_1 T_0^{1/2} + b_0$, a cubic polynomial in the variable $\sqrt{T_0}$, we obtain only one real root with an explicit formula. Squaring this formula yields
$$
T_0(X,T) = \left( \frac{(2/3)^{1/3} b_1 }{\Phi} - \frac{\Phi}{2^{1/3}3^{2/3}}\right)^2,\qquad\text{ for } \quad 
\Phi(X,T) = \left( -9b_0 + \sqrt{81 b_0^2 + 4 b_1^3} \right)^{1/3}.$$  
We insert this formula for $T_0$ into the wavenumber characteristic solution
$
    q(T)=q_0=k_{\epsilon}(T_0/\delta),
$
for $k_\epsilon(t)$ as given in equation \eqref{eq:k_eps}, obtaining the following formula for the wavenumber in terms of $X$ and $T$
\begin{equation}
    q(X,T)=k_\epsilon(T_0(X,T)/\delta).
\end{equation}
Unwinding the change of variables $(X,T) = \delta(x,t)$, and defining $t_0(x,t) := T_0(X,T)/\delta$ we can then obtain 
\begin{equation}\label{e:qt0}
    \psi(x,t) = q(X,T)=k_\epsilon(t_0(x,t)).
\end{equation}
See also Sec. 1 - 3 of \cite{howard77} for a similar derivation and discussion of the relation between the scaled and unscaled phase and wavenumbers. As noted there, we remark that, due to the uniform scaling $(X,T) = \delta(x,t)$ and the local wavenumber $q(X,T) = \psi(x,t)$, one finds that the solution $\psi(x,t) = k_\epsilon(t_0(x,t))$ can be obtained at leading order directly by solving the equivalent Burger's equation in the fast co-moving frame variables $(z,t)$
\begin{equation}
\psi_t = \partial_z\omega(\psi,t),\,\, z<0, t>0; \qquad \psi(0,t) = k_\epsilon(t),\,\, \psi(z,0) = 0
\end{equation}
using the method of characteristics, and then translating back to stationary coordinates $(x,t)$.



After setting the wavenumber to be zero ahead of the front, we obtain a prediction for the wavenumber dynamics at both the front interface as well as in the wake.  Figure \ref{f:pred_wavenumber} gives space-time diagrams of the measured wavenumber, $k_{meas}=\mathrm{Im}(A_x/A)$ (left) and predicted wavenumber, $k_\epsilon(t_0(x,t)),$ (center), as well as a plot with both laid on top of each other plotted against $x$ for several positive times (right). We find good agreement between the prediction and the measured wavenumber.  

\begin{figure}[h]
    \centering
    \hspace{-0.2in}
        \includegraphics[width=0.36\textwidth]{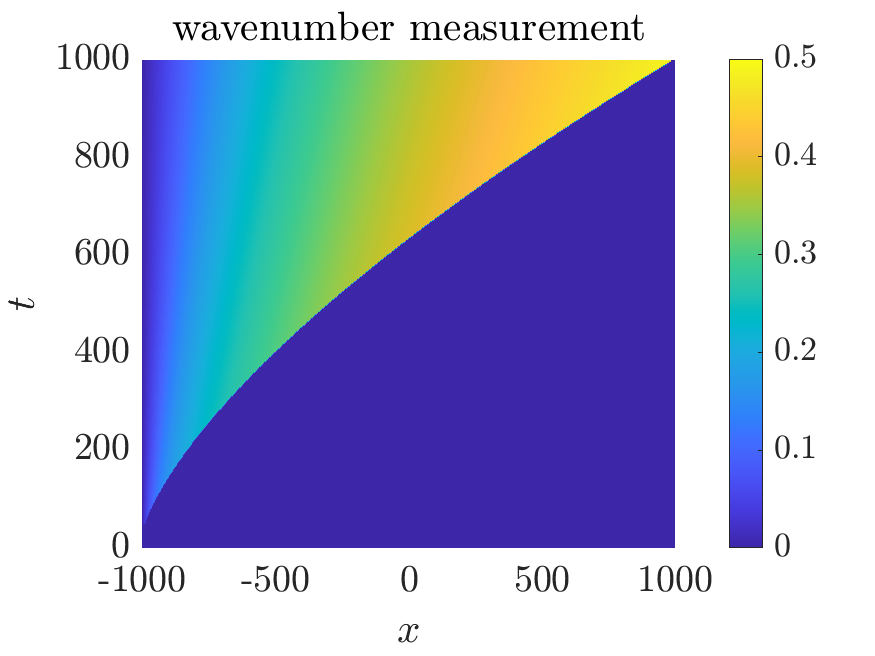} \hspace{-0.2in}
		\includegraphics[width=0.35\textwidth]{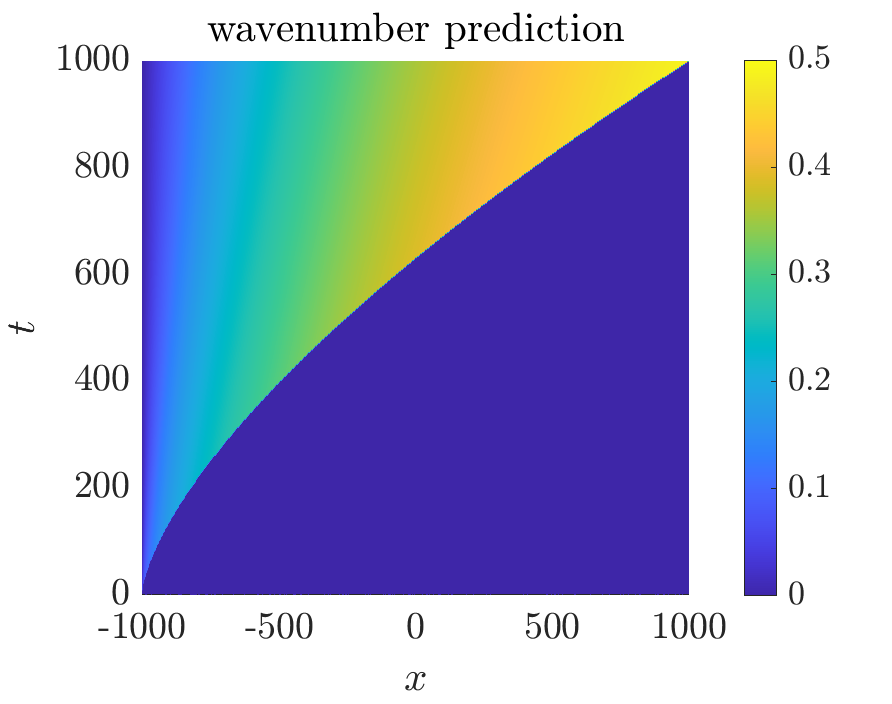}\hspace{-0.2in}
		\includegraphics[width=0.35\linewidth]{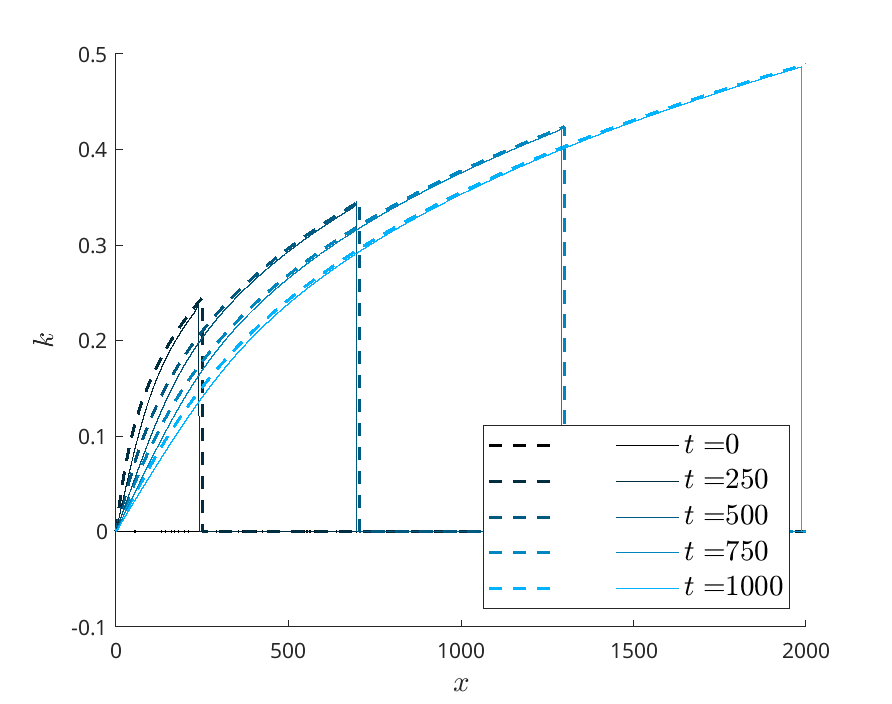}\hspace{-0.2in}
		   \caption{Comparison of the measured wavenumber $k_{meas}(x,t)=\mathrm{Im}(A_x(x,t)/A(x,t))$ and the prediction $k_\epsilon(t_0(x,t)),$ obtained via an inviscid Burger's modulational analysis; Left and center plots give spacetime diagrams (with color denoting the wavenumber) of this measurement and prediction respectively; Right plot overlays measurement (solid) and prediction (dashed), plotted against $x$ for various times $t>0$ (labeled in the figure, and varying in color); wavenumber set to zero ahead of the front interface $x_f$. }
    \label{f:pred_wavenumber}
\end{figure}

%

\subsection{Delayed Hopf bifurcation in CGL}\label{s:dhb-cgl}
As to be expected from \cite{kaper2018delayed,goh2022delayed}, the linearized space-time memory curve approach for FKPP in Section \ref{s:dhb-k} above also gives accurate predictions for the slowly varying CGL equation \eqref{e:cgl0} with $\mu = \epsilon t + \mu_0$ with $\mu_0<0$. Following the aforementioned references, we change variables to consider the $\mu$ variable as time, and consider
\begin{align}\label{e:cgl_m}
\epsilon A_\mu = (1+\ri\alpha)A_{xx} + \mu A - (1+\ri\gamma)A|A|^2.
\end{align}
Linearizing about the base state $A = 0$, we then obtain the following linear equation
\begin{equation}\label{e:lcgl_m}
\epsilon B_\mu = (1+\ri\alpha)B_{xx} + \mu B.
\end{equation}
which can once again be solved explicitly. For initial condition $B(x,0) = \re^{-x^2}$, one obtains the solution
\begin{equation}\label{e:lcgl_e}
B(x,t) = \left(4\pi(1+\ri\alpha)(\mu-\mu_0) + \epsilon \right)^{-1/2} \exp\left[ \frac{1}{2\epsilon}\left(\mu^2 - \mu_0^2 \right) -\frac{\epsilon x^2}{4(1+\ri\alpha)(\mu-\mu_0)+\epsilon} \right].
\end{equation}
With this solution, we set $|B(x,\mu)| = B_\mathrm{th}$ and take the logarithm of both sides to obtain an implicit equation for the space-time memory curve,
\begin{align}
\log(B_\mathrm{th}) &=\frac{1}{2}\log\epsilon -\frac{1}{4}\log\left( (\epsilon + 4(\mu_{mc}(x) - \mu_0))^2 + (4\alpha(\mu_{mc}(x)-\mu_0))^2 \right) \notag\\
&\qquad\qquad-\epsilon x^2 \frac{\epsilon + 4t}{(\epsilon + 4(\mu_{mc}(x) - \mu_0))^2 + (4\alpha(\mu_{mc}(x)-\mu_0))^2}.\label{e:cgl_mc}
\end{align} 
This implicit equation is solved numerically and plotted on top of the numerical solution of \eqref{e:cgl_m} in Figure \ref{f:cgl_dhb}. 
\begin{figure}[h!]    
    \centering
\includegraphics[width=.45\textwidth]{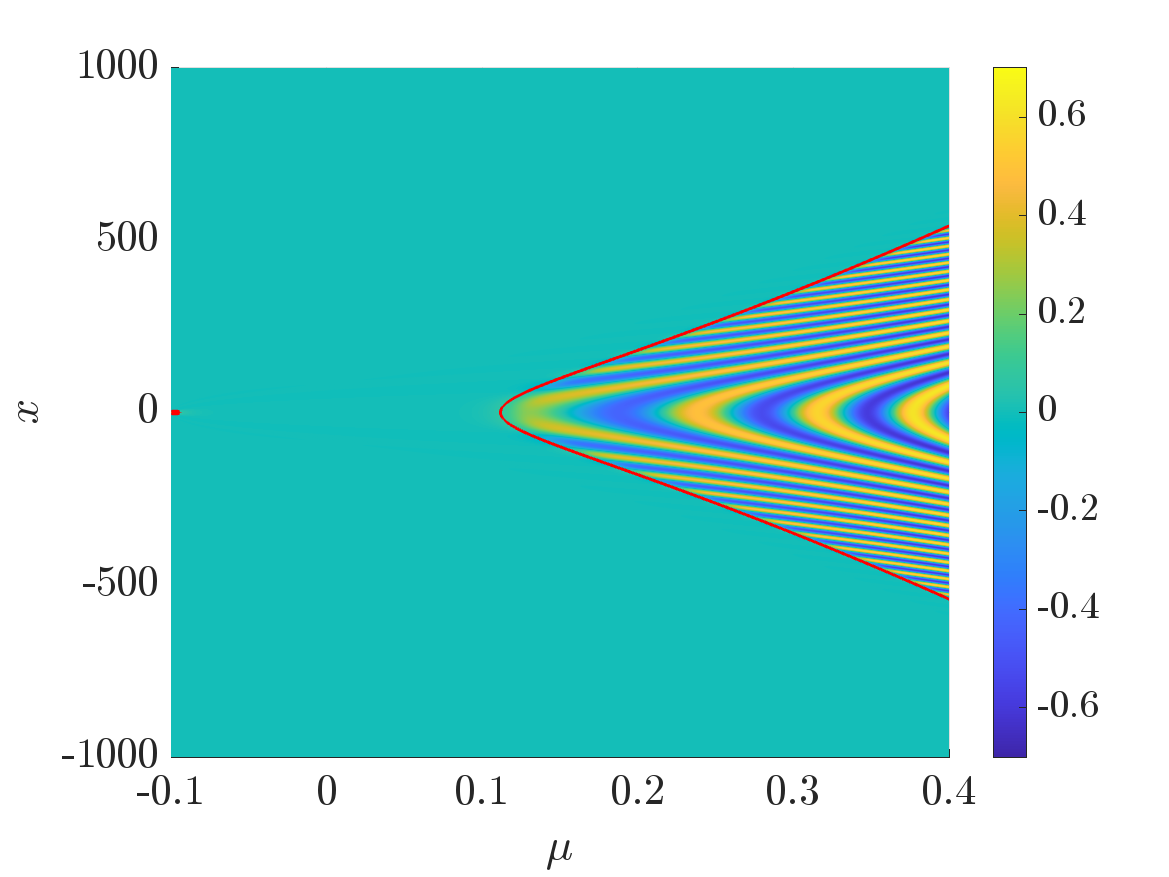}\,
\includegraphics[width=.45\textwidth]{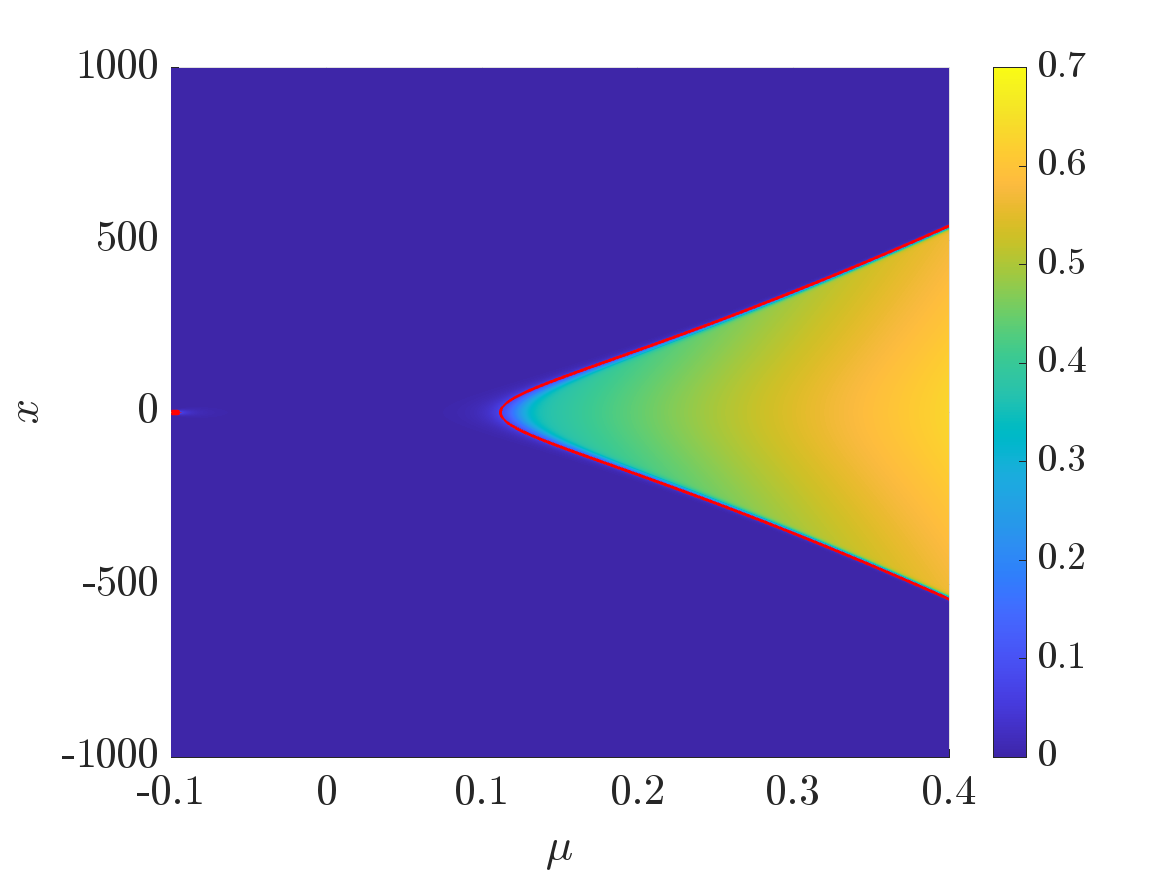}
    \caption{Depiction of delayed Hopf bifurcation and invasion in \eqref{e:cgl_m} with $\epsilon = 0.001,\mu_0 = -0.1,\alpha = 1,\gamma = 0.3, v_\mathrm{th} = 0.2$ and $A(x,0) = \re^{-x^2}$. Domain length $x\in[-1000,1000]$ with $dx = 0.1$ and $dt = 0.0025$. The red curve denotes the linear predicted space-time memory curve $\mu_h(x)$. Left plot depicts a space-time diagram of $\mathrm{Re} \, A(x,\mu)$ and the right plot gives the same for the amplitude $|A(x,\mu)|$.  }\label{f:cgl_dhb}
\end{figure}
We find good agreement between our linear prediction $\mu_\mathrm{mc}(x)$ and the spatially dependent onset of large amplitude oscillations. 


\section{Rigorous spreading in FKPP fronts}\label{s:rig}
In this section, we return to the FKPP equation \eqref{e:fkpp0} with time-dependent, increasing, unbounded, and slowly-varying parameter $\mu = \epsilon t+\mu_0$, and standard nonlinearity $f(u) = u-u^2$. We give a rigorous characterization of the spreading process in the Cauchy problem with step function initial data $u(x,0) = h_0(-x)$ with $h_0(x)$ the standard Heaviside function. Note, since $\mu(t)<0$ for $t\ll-1$, there do not exist generalized transition waves as defined in \cite{berestycki2007generalized} et. al. We show that for $t$ sufficiently large,the stable state $u\equiv 1$ invades $u\equiv0$ with the linearly predicted front position $\sigma$ derived in \eqref{e:sig} above. 
Our main result is as follows:
\begin{Theorem}\label{t:sp}
Let $\mu_0=0$ and $\epsilon>0$, and $u(x,t)$ be the solution of \eqref{e:fkpp0} with $u(x,0) = h_0(-x).$ Then, for all $\eta>0$
\begin{align}
\lim_{t\rightarrow+\infty}\sup_{x\geq \sigma(t) + \eta t}u(x,t) = 0.\label{e:sp1}\\
\lim_{t\rightarrow+\infty}\inf_{x\leq \sigma(t) - \eta t}u(x,t) = 1.\label{e:sp2}
\end{align}
\end{Theorem}
The proof of this theorem is given in the rest of this section. We first prove \eqref{e:sp1} by creating a supersolution using a solution of the linearized equation with the same heaviside initial data. We then use a difference of exponential profiles with time dependent decay rates to create a subsolution which establishes \eqref{e:sp2}. 

This result is by no means sharp, and is mostly meant to show how comparison principle methods can be applied in the unbounded parameter case, using exponentials with time dependent spatial decay rate. We note that similar spreading behavior will occur for small, bounded, and compactly supported initial data. Also, we expect results of this form will extend to more general monostable nonlinearities, e.g. $f(0) = f(1) = 0, f'(0)>0, f'(1)<0, f(u)\leq f'(0) u$ for $u\in (0,1).$ Further, we expect similar spreading behaviors to occur in other systems with monostable equilibrium configuration and an unbounded parameter which controls stability. 
\begin{Remark} 
 The results of Theorem \ref{t:sp}  also hold for $\mu_0>0$. The linearly predicted front position in this case is $\sigma_{lin}(t) = \left(4 t \int_0^{t}\mu(s)ds \right)^{1/2} = \sqrt{2\epsilon t^3}(1 + \frac{2\mu_0}{\epsilon t})^{1/2}$. One would still set the front position to be the leading order term $\sigma(t) = (2\epsilon t^3)^{1/2}$ and proceed as in the proof below. Any corrections due to the lower order terms enter in at $\mathcal{O}(t^{-1/2})$ and hence are contained within the lines $x = \sigma(t) \pm \eta t$ for any $\eta>0$  as $t\rightarrow+\infty$  
\end{Remark}

\subsection{Supersolution}\label{ss:supsol}
It will be helpful to define various quantities characterizing spreading in terms of $\mu$ such as the time-dependent speed and spatial decay rate. We have 
$
c(t) = 3\sqrt{\mu/2},\quad \nu(t) = -\sqrt{\mu/2}, 
$
so that $\nu^2 + c\nu + \mu = 0.$ Also, we define the nonlinear and linearized operators
\begin{align}
        N[u]&:=u_t-u_{xx}-\mu(t)f(u),\label{e:Nv}\\
Lv&=v_t-v_{xx}-\mu(t)v.        \label{e:Lv}
    \end{align}
We then obtain a supersolution in the following lemma.
\begin{Lemma}\label{l:sup}
Let $\bar\phi(x,t)$ be the solution of $0 = Lv$, defined in \eqref{e:Lv}, with initial condition $\bar\phi(x,0) = u(x,0)$, then $\bar\phi$ is a supersolution of \eqref{e:fkpp0}, that is $N[\bar\phi]\geq0$, and it satisfies 
\begin{equation}
\lim_{t\rightarrow+\infty}\sup_{x\geq \sigma(t) + \eta t} \bar\phi(x,t) = 0,
\end{equation}
 for any $\eta>0$.
\end{Lemma}
\begin{proof}
We first observe that $f(u)\leq f'(0) u$ for $u\geq0$ and hence $N[\bar\phi]\geq L\bar\phi = 0$, and thus $\bar\phi$ is a supersolution of the full nonlinear equation for all $t\geq0$.  The linearized equation can be solved explicitly for heaviside initial data obtaining
$$
 \bar\phi(x,t)=\frac{\exp(\epsilon t^2/2)}{2}\left[1-\mathrm{erf}\left(\frac{x+\sigma(t)}{2\sqrt{t}}\right)\right],
$$
which becomes positive for all $t>0$. Analyzing the level sets of this function yields the result. 
\end{proof}

Using this generalized supersolution, one can also obtain an upper-bound on the asymptotic tail of the front as $x\rightarrow+\infty$ by using the expansion of the (complementary) error function to find 
$$
\bar u(x,t)\sim e^{\nu(t)(x-\sigma(t))},\qquad \text{ for } x\rightarrow +\infty,\,\, t>0,
$$
where we recall that $\nu$ satisfies $\nu^2 + c\nu + \mu = 0.$ Since this supersolution is unbounded as $x\rightarrow-\infty$, and since $u = 1$ serves as a global bound for all positive solutions with data less than or equal to one, we may form the generalized supersolution $\bar{u}(x,t) = \min_x\left\{ \bar\phi(x,t),1\right\}$ which retains the same spreading and asymptotic spatial decay properties as $\bar \phi$. 

\subsection{Subsolution}\label{ss:subsol}
 We construct two subsolutions. The first will allow us to obtain tail estimates on the nonlinear solution $u$.  We then construct a refined subsolution with a $\mu$-constant maximum value (in $z$ for each $t$). Using the tail estimates of the first subsolution, we then can scale and shift the refined subsolution to bound the nonlinear solution from below and obtain a lower bound on the spreading properties of the front for sufficiently large times.

We transform the nonlinear equation into the co-moving frame $z = x-\sigma(t)$, and define $\mathcal{N}[u] = \mathcal{L}u + \mu u^2,$ with $\mathcal{L}v = v_t - v_{zz}-c(t)v_z - \mu v$. For the first subsolution, we let $v(z,t) = \bar \phi(z+\sigma(t),t)$  be the solution of the linear equation $0 = \mathcal{L} v$ defined in Section \ref{ss:supsol} above. This can be made a nonlinear subsolution for $z$ large enough by defining $\underline{w}(z,t) = A(t) v(z,t) $, and inserting into the nonlinear equation to obtain $\mathcal{N}[\underline{w}] = v\left( A' + \mu A^2 v\right).$ Next we have there exists a constant $C>0$ such that $v(z,t) \leq C e^{\nu z}$. Then, for $z\geq z_s(t):= -\sqrt{2}\mu^{-1/2}\log(\mu^{-\alpha})$ with $\alpha>0$
\begin{align}
v(z,t) \leq C \mu^{-\alpha}.
\end{align}
Choosing $A$ to solve $A' = -C \mu^{1-\alpha} A^2$, this estimate gives
\begin{align}
\mathcal{N}[\underline{w}] \leq  v\left( A' + C\mu^{1-\alpha} A^2 \right)=0, \qquad\quad z\geq z_s(t).
\end{align}
Then for $\alpha>2$ we can choose $A(0)>0$ so that $A(t)$ is bounded from above and below by positive constants for all $t\geq0.$ We also remark that $z_s(t)$ is bounded for $\mu$ large.

We next construct a subsolution which gives a lower bound on the spreading properties of the front for sufficiently large times and which can be modified and shifted to bound the nonlinear solution from below for all $z$. 
We construct this subsolution using the somewhat standard approach of taking a difference of two exponential functions with  an appropriate time-dependent shift. We choose parameters in order to keep the subsolution bounded with maximum bounded away from zero for all $t>0$, and moving with asymptotic speed $c(t)$; see \cite[Prop. 3.4]{nadin12}.  As many quantities will be dependent on $t$, we suppress this in our notation below to simplify formulas whenever convenient.  We define $\varphi(z,t) = \psi(z,t) - \theta(z,t)$ with 
$$
\psi(z,t) = e^{\nu z},\qquad \theta(z,t) = e^{\kappa(z-\delta_2)}
$$
where $\kappa = (1+C_0) \nu$ for some constant $0<C_0<1$ fixed and independent of $t$ and $\epsilon$, and the shift $\delta_2 = \delta_2(t)$ is defined as
\begin{equation}
\delta_2  := -\frac{C_2}{\mu^{1/2}},
\end{equation}
for some constant $0<C_2<C_0$ to be chosen later. Note that $\varphi$ has the same leading order decay rate as the supersolution for $z\gg1.$  The $\mu^{-1/2}$ dependent shift $\delta_2$ is chosen to balance the $\mu^{1/2}$ behavior of $\nu$ and keep the unique maximum value of $\varphi$ bounded in $\mu$. The function $\varphi$ satisfies $\varphi(z,t)\geq 0$ on the domain $z\in[z_0(t),\infty)$ with 
$$
z_0(t) := \frac{ \kappa \delta_2 }{C_0 \nu} = -\frac{(1+C_0)C_2}{C_0\mu^{1/2}}.
$$
Further, it has a single critical point 
$$
z_* = \frac{\kappa \delta_2 + \log(\kappa/\nu)}{ \kappa-\nu} 
$$
with $\mu$-independent maximum value
\begin{align}
\varphi(z_*,t) &= (\kappa/\nu - 1) \exp\left[\frac{\kappa}{\nu - \kappa}\left(-\nu \delta_2 +  \log(\kappa/\nu)\right) \right]\notag\\
&= \frac{C_0}{(1+C_0)^{1+1/C_0}} \exp\left[\frac{C_2(1+C_0)}{\sqrt{2}C_0}\right],\label{e:vpzs}
\end{align}
and is monotonically decreasing for $z\geq z_*.$ Inserting $\varphi$ into the PDE, and observing that $(\partial_z^2 + c\partial_z + \mu)\psi = 0$ by the definition of $\nu$, we obtain
\begin{align}
\mathcal{N}[\varphi] &= \varphi_t + p(\kappa) \theta + \mu \varphi^2\notag\\
&= z \nu' \psi - (z-\delta_2)(1+C_0)\nu'\theta  + \kappa \delta_2'\theta + p(\kappa) \theta + \mu \varphi^2\notag\\
&=:I + II
\end{align}
where we define $p(\kappa) := (\kappa^2+c\kappa + \mu) = -\nu^2 C_0(1-C_0)$ in the first line, and set
$$
I = z \nu' \psi  - (z-\delta_2)(1+C_0)\nu'\theta + p(\kappa)\theta/2 ,\quad
II =p(\kappa)\theta/2 + \mu\varphi^2 + \kappa \delta_2' \theta\notag\\
$$
Here we decompose $\mathcal{N}[\varphi]$ into two groups $I$ and $II$ and separately bound both from above by 0. We have split the $ p(\kappa)\theta$ term to use its strict negativity for $C_0\in(0,1)$ in both groups of terms.

 We begin by showing $II\leq 0$. We first note that 
 \begin{equation}\label{e:kapdel1}
     \kappa \delta_2'\theta = -\frac{(1+C_0) C_2}{2\sqrt{2} \mu}\theta < 0 
 \end{equation} for all $t>0$. Hence it is sufficient to show that $p(\kappa)\theta/2 + \mu\varphi^2\leq0$ for suitably chosen $C_0$ and $C_2$ after possibly scaling $\varphi$ by some constant independent of $\mu$. We note straight-forward computation gives $\varphi(z_*) = C_0\theta(z_*)$. Also, for $C_2<\sqrt{2}\log(1+C_0)$, we have that  $\varphi(z_*) \leq C_0$ by \eqref{e:vpzs}. We also compute 
 $$
\varphi(z_*) \leq \varphi(z_*)/C_0 = \theta(z_*). 
$$
We then claim that $\varphi(z,t)^2 < C_0^2 \theta(z,t)$ for all $z\geq z_0.$ Momentarily assuming this claim, we then have that 
\begin{align}
\frac{p(\kappa)}{2} \theta + \mu \varphi^2 &= -\frac{\mu C_0}{4}(1-C_0) \theta + \mu \varphi^2\notag\\
&\leq \mu \left( -\frac{1}{4}C_0(1-C_0) +C_0^2\right)\theta\notag\\
&= \mu\left( -\frac{1}{4}C_0 + \frac{5}{4} C_0^2\right)\theta \leq 0\label{e:pkap2}
\end{align}
choosing $C_0$ sufficiently small - in particular $C_0 \in(0, 1/5).$ To prove this claim, we note the $z$ decay properties of $\psi$ and $\theta$ imply that $\varphi^2$ decays faster than $\theta$ for in $z\gg1$ and further that there exists a $z_2>0$ bounded in $\mu$ such that 
$$
\varphi^2\leq C_0^2 \theta, \qquad z\geq z_2. 
$$
Since $\varphi(z_*)$ is independent of $\mu$, and due to the quadratic dependence of the left-hand side, one can then obtain the same inequality on the uniformly bounded domain $(z_0,z_2)$ by rescaling $\varphi$ to be sufficiently small.  
All together, combining \eqref{e:kapdel1} and \eqref{e:pkap2} we then have that $II\leq 0.$


We now bound $I$. Define $B(z,t) := \nu' \psi - \nu'(1+C_0)(z - \delta_2)\theta.$  The general approach will be to first show that $B\leq0$ for all $z\geq z_1$ for some $z_1$ which is bounded uniformly in $\mu$. Then, we find that $ B\lesssim \mu^{-1/2}$ on the interval $z\in[z_0,z_1]$. This allows us to conclude, on the region where $B(z)>0$, that $B$ is dominated by $\frac{p(\kappa)\theta}{2}=-C_1\mu$  for $\mu$ sufficiently large.

\begin{Lemma}\label{l:bz}
On the interval $(z_0,\infty)$, $B(z,t)$ has one zero, $z_1$, satisfying $0<z_1<  C_3 \mu^{-1/2}$ with $C_3 = C_2\left(e^{\frac{(1+C_0) C_2}{\sqrt{2} }} - 1\right)^{-1}$, $B>0$ on $(z_0,z_1)$, and $B<0$ for $z\in(z_1,\infty).$
\end{Lemma}
\begin{proof}
First, as $\nu'<0$, we set $B = \nu' \tilde B$ with $\tilde B = z\psi - (1+C_0)(z-\delta_2)\theta$ and prove the analogous results for $\tilde B$. Next, we observe that our choice of $\delta_2$ gives, via a brief calculation, that $ \tilde B(z_0)=0$, and $\tilde B(z)<0$ for $z\in(z_0,\delta_2]$. We study zeros on $z>\delta_2$ by rearranging $\tilde B=0$ as
\begin{align}\label{e:z1p}
   (1+C_0)^{-1}\frac{z }{z-\delta_2} = \theta(z)/\psi(z). 
\end{align}
We observe the left hand side of \eqref{e:z1p} has a vertical asymptote at $z=\delta_2$, goes through zero at $z=0$, and has a horizontal asymptote to $(1+C_0)^{-1}$ as $z\rightarrow+\infty$. The right hand side, $\theta/\psi = e^{(1+C_0)\nu(z - \delta_2) - \nu z }$ is positive, monotonically decreasing, and hence bounded from above for $z>\delta_2$ by $\theta(\delta_2)/\psi(\delta_2)<1$. Further, by the monotonicity and asymptotics of each function we see that there is one root, $z_1$, of \eqref{e:z1p} in $[0,\infty)$. The sign properties of $\tilde B$ follow.    To bound $z_1$, it also follows that the zero $z_1$ lies to the left of the root, $\tilde z_1$, of the equation 
\begin{equation}\label{e:zt1p}
    \frac{z }{z-\delta_2} = r
\end{equation} 
with $r := \frac{(1+C_0)}{2}\theta(0)/\psi(0)= \frac{1+C_0}{2}e^{-(1+C_0)C_2/\sqrt{2}}$. The choice of $C_0$ gives that, for all $C_2>0$ sufficiently small, $0<r<1$.  We then compute
\begin{align}
    \tilde z_1 &= \frac{r \delta_2}{r-1} 
    = \frac{C_2 r}{1 - r} \mu^{-1/2}.\notag
\end{align}
The constant $C_3 = \frac{C_2 r}{1-r}$ multiplying $\mu^{-1/2}$ in the last line, is independent of $\mu$, and can be chosen to be positive for some $C_2>0$ small. This bounds the root $z_1$ from above by $C_3\mu^{-1/2}$ and hence gives the result of the lemma.
\end{proof}
From this lemma we have that $B\leq0$ for all $z\geq z_1$, leaving a bounded domain $[z_0,z_1]$ where $B\geq0.$  On this interval, direct computation then gives that $|B(z)|\leq \epsilon C_4 \mu^{-1/2}$ for some $C_4>0$ independent of $\mu$. One can also estimate 
\begin{equation}\label{e:minth}
\min_{z\in[z_0,z_1]}\theta(z) > \theta(\tl z_1) = \exp\left( -\frac{(1+C_0)C_2}{\sqrt{2}(1-r)}\right) > e^{-1}
\end{equation}
where we have used the fact that $-s/(1-\frac{1+C_0}{2}e^{-s}) \approx -2s/(1-C_0)$ for $s\approx0$ and set $s = \frac{(1+C_0)C_2}{\sqrt{2}}$. The last inequality follows from this approximation and the fact that $C_0\in(0,1/5)$ by our previous assumption.

We can then obtain $I\leq0$ for $z\in(z_0,z_1)$
\begin{align}
\frac{p(\kappa)}{2}\theta &= -\frac{C_1}{2}\mu\theta(z,t) \notag\\
&\leq -\frac{C_1}{2} \mu \min_{z\in[z_0,z_1]}\theta(z,t) \notag\\
&< -\frac{C_1}{ 2e } \mu  \notag\\
&\leq -B(z),
\end{align}
for $\mu$ sufficiently large, in particular $\mu\gtrsim \epsilon^{2/3}$ up to some constant dependent on $C_1$ and $C_4.$ Here the third line follows from \eqref{e:minth} and the last inequality by the uniform bound on $|B(z)|$ for $\mu$ sufficiently large. Hence we may conclude $B(z) + \frac{1}{2} p(\kappa)\theta \leq 0$ for all $z\geq z_0$.  Combining this with the bound on $II$ we then obtain $N[\varphi]\leq0$ as desired. We then define the generalized subsolution $\underline\varphi(z,t) =\max\{0,\varphi\}$ so that the function is strictly non-negative.

\subsection{Proof of Theorem \ref{t:sp}}
Before completing the proof, we briefly comment on the application of maximum and comparison principles in our case. Say we have $u(x,0) \leq v(x,0)$ with $N[u(x,t)] = 0$, $N[v(x,t)]\geq0,$ and we have apriori bounds $0\leq u,v\leq 1$ for $t\geq0$. To derive a comparison principle, we define $w = v - u$ so that $w(x,0)\geq 0$ and 
$
w_t -\left( w_{xx} + \mu w\left(1- (u+ v)\right) \right) \geq0.
$
We then define $\tilde w = \exp\left(\int_{T_0}^t \mu(s) ds\right) w$ for some initial time $T_0\geq0$ so that 
$$
\tilde w_t -\left( \tilde w_{xx} - \mu \tilde w (u+v)\right) \geq0,
$$
which has a linear coefficient $-\mu(u+v)$ on the $\tilde w$ term which is bounded from above due to the apriori bounds on $u$ and $v$ and the fact that $\mu(t)$ is bounded from below. We can then apply standard maximum principles to obtain that $\tilde w(x,t)\geq0$ and hence that $u(x,t)\leq v(x,t)$ for all $x$ and $t\geq T_0.$ A similar argument can be used to obtain a comparison principle for subsolutions. 

Since $\bar u(x,0) =u(x,0)$, the above maximum principle gives that $u(x,t)\leq \bar u(x,t)$ for all $t>0$. We then readily conclude the upper bound on the spreading location \eqref{e:sp1} from the construction of the supersolution $\bar u(x,t)$ in Lemma \ref{l:sup}. 

To conclude the lower bound, we further alter the subsolution $\underline{\varphi}$ following the approach outlined by \cite{hamel2013}. In particular, we note that all solutions of \eqref{e:fkpp0} with non-negative initial data bounded above by 1,  converge to $u\equiv1$ as $t\rightarrow+\infty$, locally uniformly in $x$. Letting $\gamma:=\underline{\varphi}(z_*,t)$, the $\mu$- and hence $t$-independent maximum value of the subsolution $\varphi$ defined in \eqref{e:vpzs}, there exists a time $T_0$ such that $u(0,t)\geq \gamma $, for all $t\geq T_0$. Since the maximum principle readily gives that $u(x,t)$ is monotonically decreasing in $x$ for all $t>0$, we then have $u(x,t)\geq \gamma$ for all $x\leq 0$ and $t\geq T_0$.  We then modify the generalized subsolution, 
$$
\underline{\varphi}_\gamma(z,t) = \begin{cases}
    \gamma,\qquad&z<z_*\\
    \underline{\varphi}(z,t),\qquad& z\geq z_*.
\end{cases}
$$

Next, since the preliminary subsolution $\underline{w}$ defined at the start of Section \ref{ss:subsol} implies that $u$ has the the same tail asymptotics as $\underline{\varphi}_\gamma$, in particular $u(x,t)\geq C e^{\nu (x-\sigma(t))}$ for $x\geq \sigma(t)+z_s(t)$, there exists a shift $x_1$, scaling constant $\beta>0$, and time $T_1\geq T_0$ sufficiently large so that Lemma \ref{l:bz} applies and
$$
u(x,T_1)\geq  \beta\underline{\varphi}_\gamma(x - \sigma(T_1) + x_1,T_1),
$$ for all $x$. Hence, as $\beta \underline\varphi_\gamma$ is also a subsolution,  the comparison principle gives
$$
u(x,t) \geq \beta\underline{\varphi}_\gamma(x - \sigma(t) + x_1,t), \qquad \text{ for } t\geq T_1. 
$$
Since $z_*$ decays to zero like $\mu^{-1/2}$ for $t$ large, we have for any $\eta>0$ that $x = \sigma(t)+z_*(t) \geq  \sigma(t)-\eta t$ for $t$ sufficiently large. This implies that $\underline\varphi_\gamma(x-\sigma(t)+x_1,t)\equiv \gamma$ along any ray $\{x = \sigma(t) -\eta t\}$ for $t$ sufficiently large. Pairing this with the locally uniform convergence of solutions $u$ with $0<u\leq1$ to $1$ as $t\rightarrow+\infty$, we obtain the lower bound \eqref{e:sp2},  completing the proof of Theorem \ref{t:sp}.

\section{Other pattern-forming equations and discussion}\label{s:disc}
We have studied front invasion into an unstable state in the presence of an increasing, slowly-varying parameter in both the case where the front leaves behind a spatially homogeneous state as well as the case where it leaves behind a locally periodic pattern. In both cases, a linearized analysis accurately predicts the nonlinear front position and spatial decay asymptotics of the leading edge and shows that fronts accelerate faster than the speed predicted by a frozen-coefficient analysis. We consider the former in the prototypical FKPP equation to introduce our approach and also establish a rigorous spreading result in the previously unstudied case (at least to our knowledge) of an unbounded parameter. We consider the latter case of patterned-fronts in the CGL equation, once again characterizing the accelerating invasion, but also using the leading-edge temporal oscillation frequency of the linearized equation to predict the local spatial wavenumber established at the leading edge. This leading-edge prediction then provides a dynamic Dirchlet boundary condition for an inviscid Burger's equation which accurately predicts the slowly-varying wavenumber dynamics in the wake of the front. In both cases, when the dynamic parameter has $\mu(0)<0$ so that the trivial state being invaded is initially stable, we employed the recently developed concept of a \emph{spacetime memory curve} to predict the spatially dependent delayed onset of instability and hence delayed invasion. 

The approaches for patterned fronts with a dynamic parameter developed here can, in principle,  be applied to other pattern forming systems, such as the Swift-Hohenberg equation, Cahn-Hilliard equation, or a reaction-diffusion system, where pattern-forming invasion into an unstable homogeneous state has been shown to exist for static parameters. We expect the introduction of a dynamic parameter would once again lead to an accelerating invasion front as well as a slowly varying-wavenumber. 
\begin{figure}[h!]    
    \centering
\includegraphics[width=.45\textwidth]{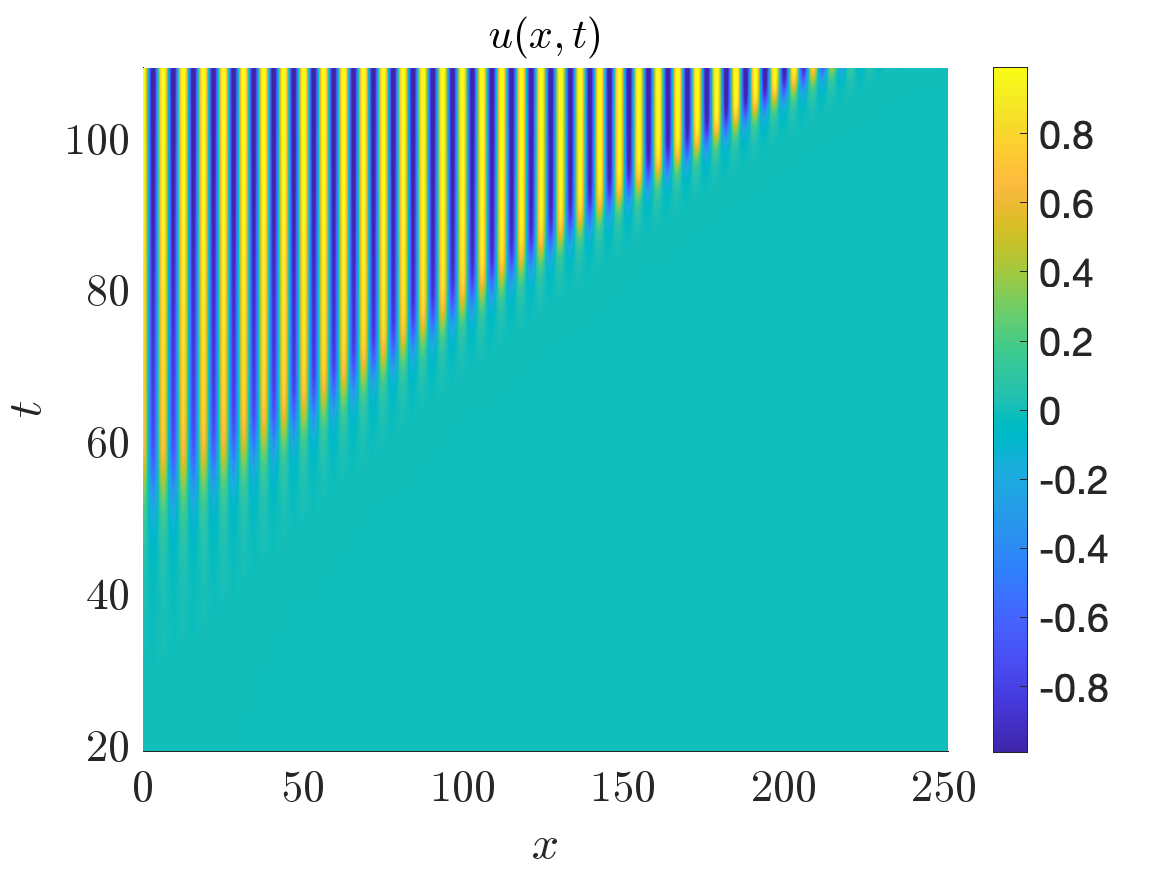}\,
\includegraphics[width=.45\textwidth]{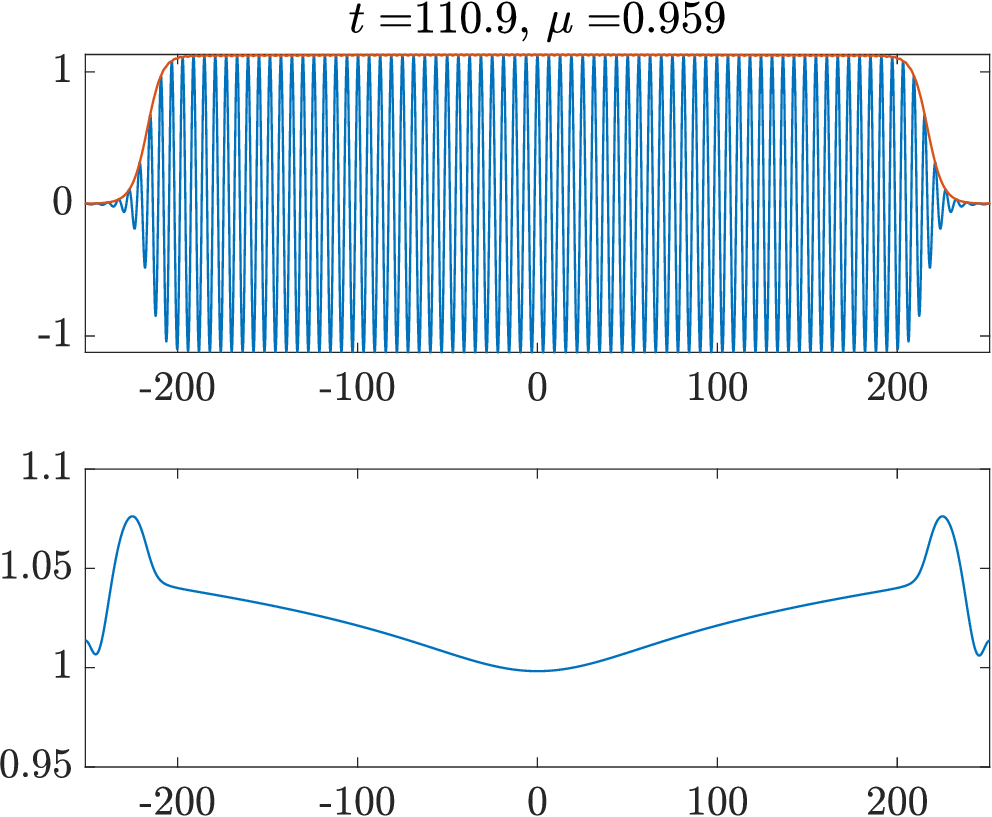}\\
\includegraphics[width=.45\textwidth]{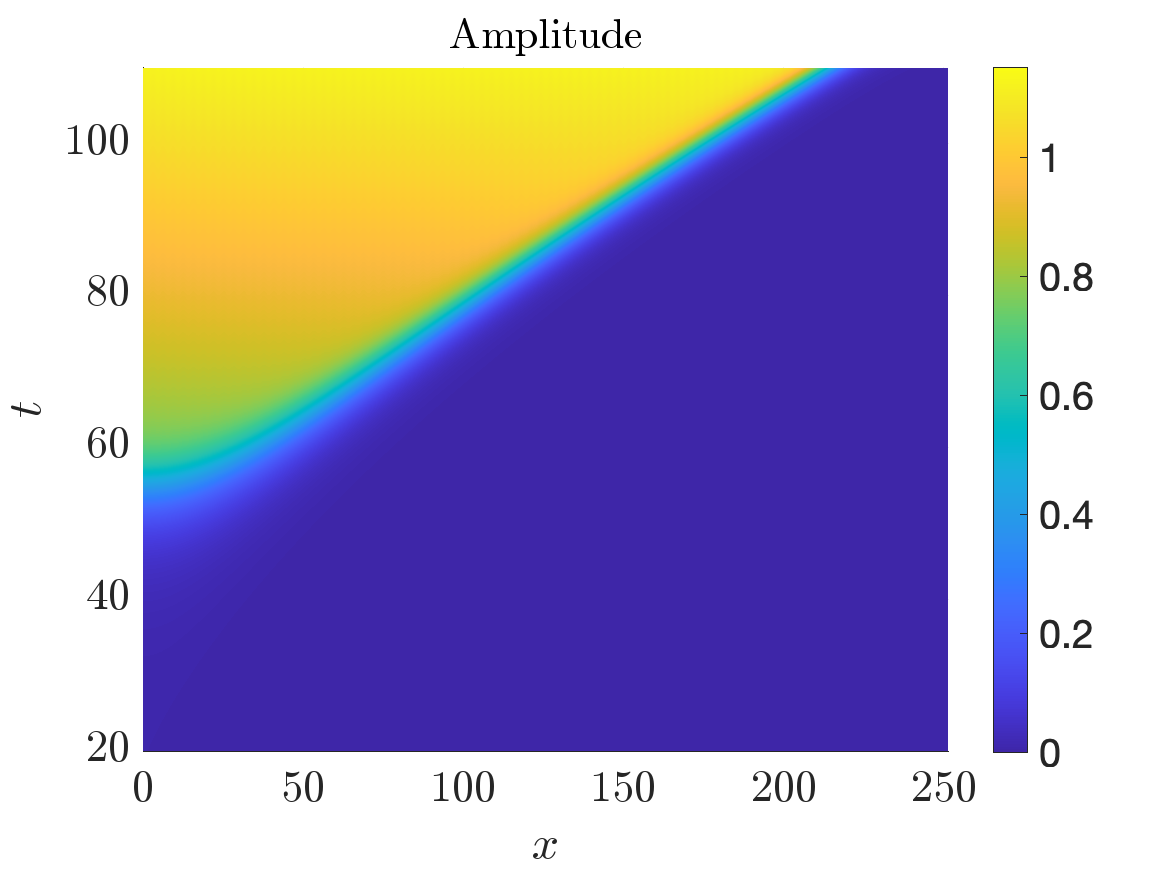}\,
\includegraphics[width=.45\textwidth]{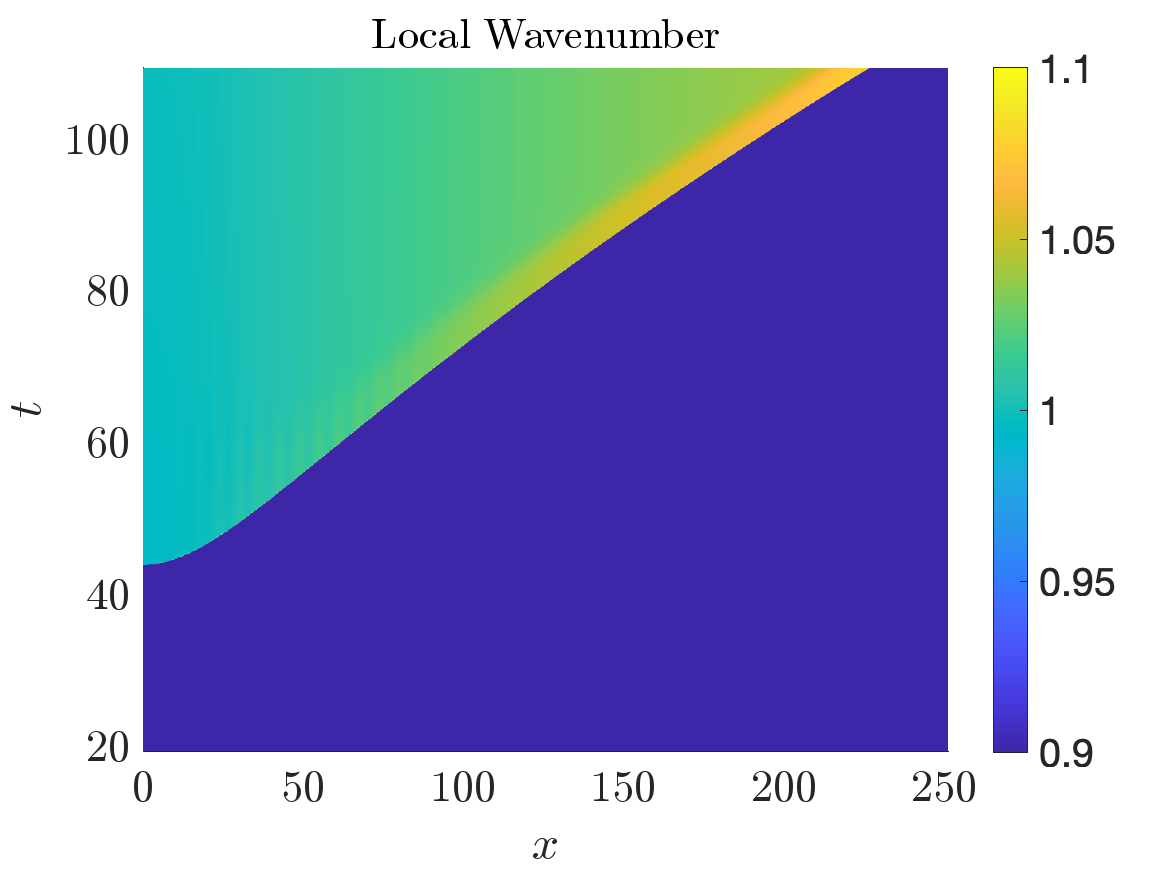}
    \caption{Delayed bifurcation and patterned invasion into an unstable state in \eqref{e:sh} with $\mu_0 = -0.15$ and $\epsilon = 0.01$. The equation used spectral discretization in space on the domain $x\in[-80\pi,80\pi]$ with $N = 2^{13}$ modes and 4th order Runge-Kutta Exponential Time differencing \cite{kassam05} with $dt = 0.001$ and Gaussian initial data $u(x,0) = \re^{-x^2}$. Upper left: Spacetime diagram of solution; Upper right: depiction of the solution profile (top)  $u(x,t)$ (blue) along with its amplitude (orange) and local wavenumber (bottom) at time $t = 110.0$; In the bottom row, spacetime diagrams of the amplitude and local wavenumber are plotted. Wavenumber measured using the iterative Hilbert transform approach; see  \cite{dunn2024transverse} for an example of this. }\label{f:sh}
\end{figure}
\begin{figure}[h!]    
    \centering
\includegraphics[width=.5\textwidth]{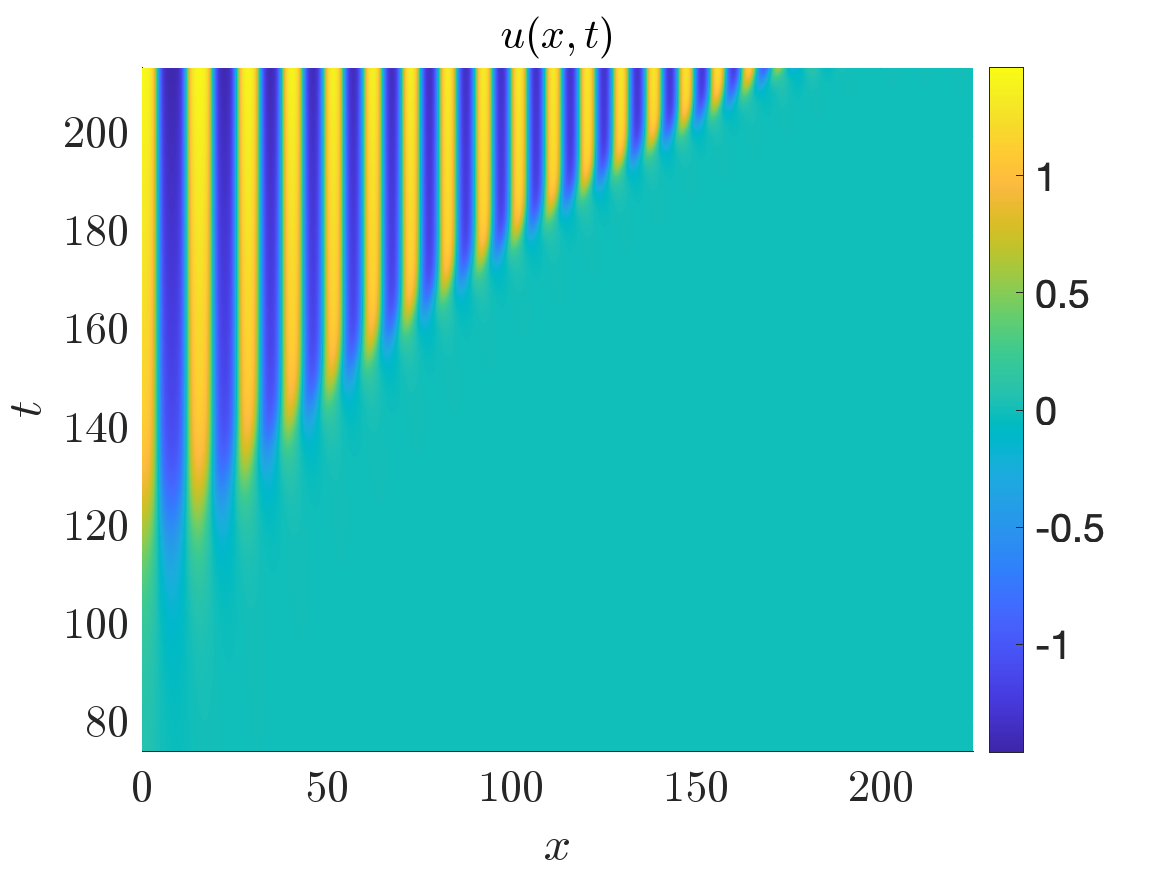}\,
\includegraphics[width=.4\textwidth]{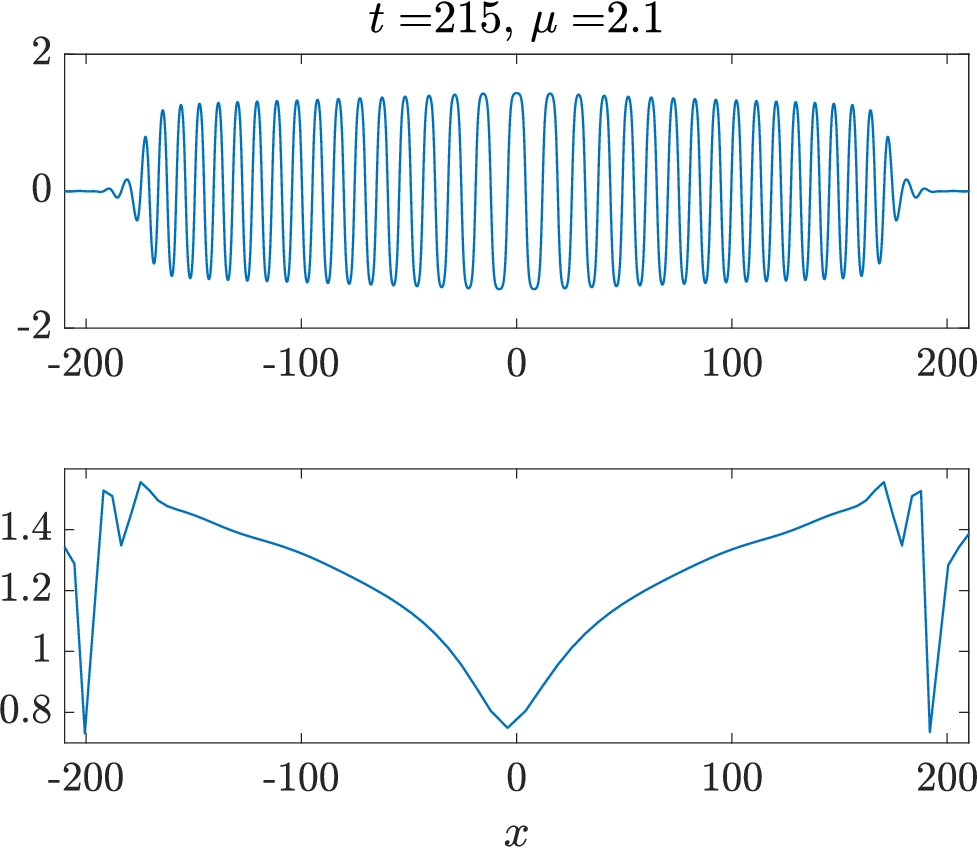}\\
    \caption{Delayed bifurcation and patterned invasion into an unstable state in \eqref{e:ch} with $\mu_0 = -0.05$ and $\epsilon = 0.01$. The same simulation approach, parameters, and initial condition as in Figure \ref{f:sh} were used.  Left: Spacetime diagram of solution; Right: depiction of the solution profile $u(x,t)$ (top) and local wavenumber (bottom) at time $t = 215$; Wavenumber determined by measuring distances between consecutive zeros of the solution, note variations occur as front tail gets close to zero.  }\label{f:ch}
\end{figure}
Figures \ref{f:sh} and \ref{f:ch} respectively provide preliminary numerical simulations of the Swift-Hohenberg and the Cahn-Hilliard equations with dynamic parameter $\mu(t)$ with $\mu_0$:
\begin{align}
u_t &= -(1+\partial_x^2)^2 u + \mu(t) u - u^3,\label{e:sh}\\
u_t &= -\partial_x^2\left(\partial_x^2 u + \mu(t)u - u^3 \right).\label{e:ch}
\end{align}
We observe accelerated invasion with non-constant front speed, and a spatially varying non-constant wavenumber selected in the wake. Indeed, we find as the speed increases, the wavenumber selected at the front interface also increases.  For the Swift-Hohenberg equation \eqref{e:sh}, the amplitude and local wavenumber of the patterned-front was measured using the iterative Hilbert transform approach; see \cite{dunn2024transverse} for the detail of this computational technique. While wavenumber variations here are relatively small due to the strong preference for wavenumbers close to 1 for $\mu$ close to zero, we indeed find the the invasion front leaves a slowly-increasing wavenumber as it increases. A plot of the solution (blue) and its amplitude (orange) for a fixed $t$ is given in the upper frame of the top right plot in Figure \ref{f:sh}, while the local wavenumber is given in the bottom frame.  Further a spacetime diagram of the wavenumber dependence for a range of times $t$ is given in the bottom right plot. In this last plot, we have thresholded the wavenumber profile, setting it equal to zero whenever the amplitude of the front was less than $0.05$. The bottom left plot of Figure \ref{f:sh} also depicts the amplitude of the front which, since $\mu_0<0$, undergoes spatio-temporal delayed onset and invasion just as in the FKPP and CGL equations above. 

We observe similar behaviors in the Cahn-Hilliard equation \eqref{e:ch}, where a Gaussian perturbation in the center of the domain initially leads to the (delayed) formation of a large kink/anti-kink solution. A front forms where further kinks and anti-kinks are layed down, with the distance between peaks narrowing as time evolves and $\mu$ increases. We note that here the wavenumber was determined by measuring the distance between consecutive zeros of the solution for each fixed $t$.  

 We expect that the linearized equation about $u = 0$ -   which is $v_t = -(1+\partial_x^2)^2 v + \mu(t) v$ in the case of \eqref{e:sh} - will give leading-order predictions for patterned-invasion in both equations.  If an explicit solution for $v$, as in the diffusive scalar case, is not attainable we expect that one could obtain a leading order diffusive expansion about the critical mode $e^{ix}$ via its pointwise Green's function \cite{holzer2014criteria}, from which front position $\sigma(t)$ and frequency $\omega(t)$ could be obtained as above. As we expect the oscillatory front interface to conserve nodes \cite{van2003front}, the frequency and speed $c(t) = \sigma'(t)$ at the leading-edge would then give the local selected wavenumber $k(t) = \omega(t)/c(t)$. From the standpoint of delayed bifurcation, it would be interesting to apply the recently developed geometric blow-up techniques of  \cite{hummel2025geometric,jelbart2024formal}, which obtain Ginzburg-Landau modulational equations for the solution in various blow-up charts near $(u,\mu )= 0$,  in order to more precisely characterize the spatio-temporal bifurcation delay; see also \cite{avitabile2020local}. More generally, it would be of interest whether a time-dependent modulation equation could be derived which not only predicts amplitude and front interface but also the selected wavenumber for a long but finite time interval. As a rigorous proof of pattern selection by invasion into an unstable state for a even static parameter $\mu = \mu_0>0$ is still open for general systems (see \cite{avery2024sharp,avery2024front,avery2023stability} for related works), we expect rigorous existence and selection proofs to be difficult.

It would also be of interest to consider dynamic parameters in systems with \emph{pushed} fronts \cite{van2003front}, where the front interface spreads with speed faster than predicted by the linearized equation. One natural place to begin this study would be the Nagumo-scalar reaction diffusion equation 
\begin{align}\label{e:nag}
u_t = u_{xx} +u(1-u)(\mu(t)+u)
\end{align}
where $u=1$ invades $u = 0$. When $\mu$ is constant, it is known that invasion fronts are pushed for $\mu\in(0,1/2)$ and travel with the speed $c_p(\mu)=\sqrt{2}(1/2+\mu)$, invading with speed faster than predicted by the linearized equation. For $\mu\geq 1/2$, the selected front is pulled, traveling with the linearly predicted \emph{pulled} speed $c_{lin}(\mu) = 2\sqrt{\mu}$; see \cite{avery2023pushed} and references therein for a summary of these results. Figure \ref{f:nag} gives the results of direct numerical simulation of \eqref{e:nag} with linear ramp $\mu(t)$ which slowly ramps through the pushed regime and then into the pulled regime for larger $\mu$.  See also Example 2 of \cite{nadin12} which gives an explicit formula for the front in the case of $\mu(t)>1$.
\begin{figure}[h!]    
    \centering
    \includegraphics[width=.46\textwidth]{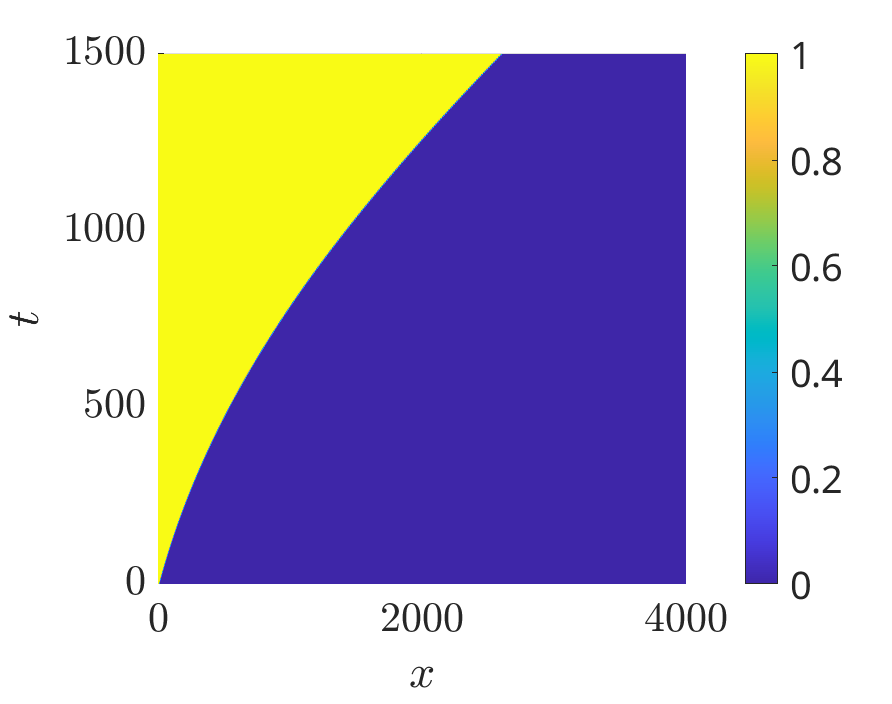}\,
\includegraphics[width=.4\textwidth]{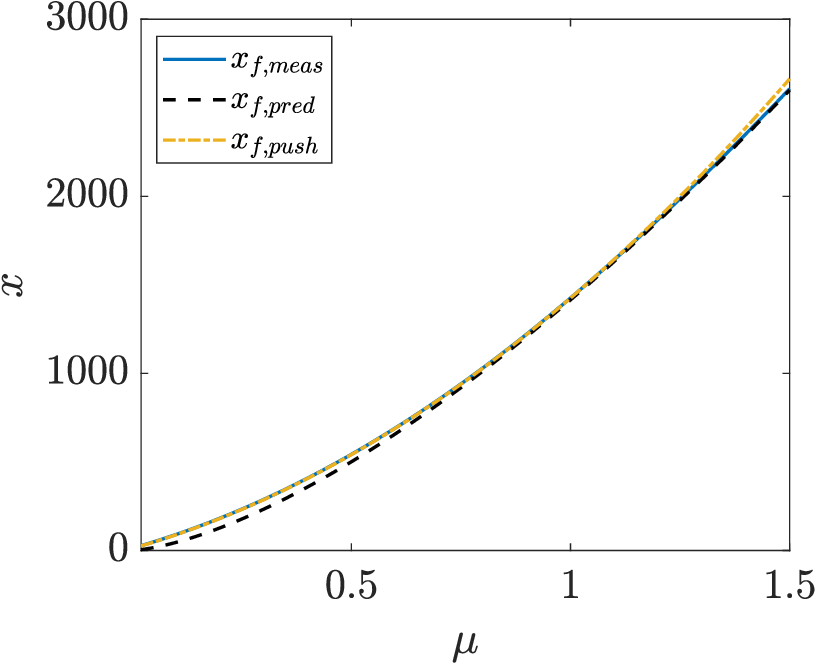}\,
\includegraphics[width=.4\textwidth]{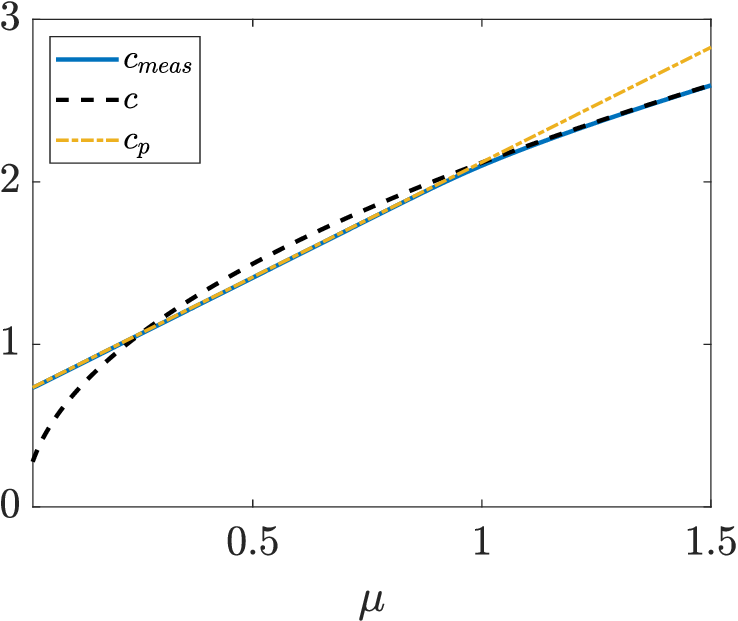}\,
\includegraphics[width=.4\textwidth]{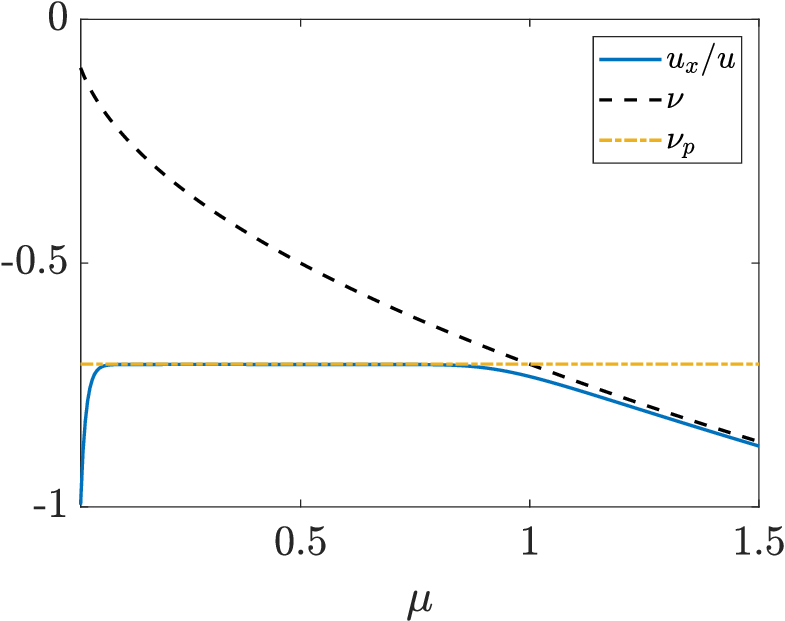}\,
    \caption{Direct numerical simulations of \eqref{e:nag} with $\mu_0 = 0$ and $\epsilon = 0.001$, using same numerical approach and parameters as in Fig. \ref{f:front_comp}. Top left: Space time diagram of solution; Top right: comparison of measured front position (blue) to dynamic linear prediction $x_{f,pred}$ (dashed black) defined in \eqref{e:sig}, and frozen coefficient pushed prediction $x_{f,p}$ (dot-dashed yellow) defined below. Bottom left: Comparison of the measured instantaneous front speed (blue) with the dynamic linear prediction (dashed black) and frozen-coefficient pushed speed (yellow dot-dashed); Bottom right: comparison of the leading-edge decay of the front (measured using $u_x/u$) with the dynamic linear prediction $\nu(\mu)$ and frozen-coefficient pushed prediction $\nu_p(\mu)$. }\label{f:nag}
\end{figure}

Interestingly, and in contrast to the FKPP equation above, we observe that the front initially invades with the frozen-coefficient speed $c_p(\mu(t))$ defined for the constant coefficient equation; see Figure \ref{f:nag} top left. Integrating this frozen-coefficient speed, then gives an accurate prediction for the front position $x_{f,p} := \int_0^t c_p(s) ds.$; see Figure \ref{f:nag} top right. In this regime, we observe that the leading-edge of the front exhibits the typical steep decay of pushed fronts, behaving like $u(x,t)\sim e^{\nu_p x}$ where $\nu_p$ is a root of the linear dispersion relation $\nu^2 + c_p(\mu) \nu + \mu$.  Interestingly, we find for $\mu\in[0,1/2]$ the front selects the most negative root $\nu_p = -c_p/2 -\sqrt{c_p^2/4 - \mu}$ while for $\mu$ above 1/2, roughly until 0.8 it selects the less negative root $-c_p/2 +\sqrt{c_p^2/4 - \mu}$. We expect this corresponds to the frozen-coefficient pushed-pulled transition at $\mu = 1/2$ where the traveling wave heteroclinic of $0 = u_{zz} + c_p(\mu) u_z + u(1-u)(\mu+u)$ converges into the origin along the weak stable eigendirection instead of the strong stable one; see also Fig. 5 of \cite{goh11}.

As $\mu$ increases we observe a transition regime for $\mu$ approximately between 0.8 and 1 where the front speed transitions to the dynamic linearly predicted speed $c(\mu) = \frac{3}{2}\sqrt{2\mu}$ derived in \eqref{e:ct} - \emph{not} the instantaneous speed $c_{lin} = 2\sqrt{\mu}$ - and the front position moves with the corresponding linear prediction $x_{f,pred} = \sigma(t)$  in \eqref{e:sig}.  Further, as seen in Figure \ref{e:nag} bottom right, the leading-edge front decay rate transitions in this regime between $\nu_p$ to the dynamic linear prediction $\nu(\mu) = -\sqrt{\mu/2}$ derived in \eqref{e:nu} above. We particularly note that the dynamic transition, denoted as $\mu_*$, from pushed to pulled invasion is delayed, happening at $\mu$ values (roughly 0.8) larger than the constant-coefficient pushed-pulled transition $\mu = 1/2.$  It would be of interest to determine the dependence of the delay $\mu_* - 1/2$ as $\epsilon\rightarrow0^+.$ Further, we also note the possible connection to the ``semipushed" transition regime discussed in \cite{birzu2018fluctuations}.



In a pattern-forming system, one might consider subcritical versions of the complex-Ginzburg Landau $A_t = (1+i\alpha) A_{xx} + \mu(t) A + (\rho + i\gamma) A|A|^2 - (1+i\beta)A|A|^5$,  Swift-Hohenberg $u_t = -(1+\partial_x^2)^2 u + \mu(t) u +\gamma u^3 - u^5$, or Cahn-Hilliard equations $u_t = -\left(u_{xx} + \mu(t)u + \gamma u^3 - \beta u^5 \right)_{xx}$, with $\mu$ increasing from zero for suitably chosen constants $\gamma, \beta>0$. Here, we expect fronts to initially resemble the pushed-type fronts until $\mu$ passes through a critical value $\mu_*$ where constant-coefficient fronts become pulled.

\bibliographystyle{siam}
\bibliography{main}

\end{document}